\documentclass[11pt,letterpaper]{article}
\usepackage[utf8]{inputenc}
\usepackage{latexsym}
\usepackage{amsmath}
\usepackage{amssymb}
\usepackage{amsthm}
\usepackage{epsfig}
\usepackage{hyperref}
\usepackage{subfig}

\usepackage{tikz}

\usepackage{algorithm}
\usepackage{algpseudocode}
\usepackage{xcolor}

\usepackage{dsfont}
\usepackage{tcolorbox}
\usepackage{enumerate}

\usepackage[capitalize, nameinlink]{cleveref}

\usepackage[square,numbers,sort]{natbib}
\newtheorem{theorem}{Theorem}

\newtheorem{lemma}[theorem]{Lemma}

\newtheorem{fact}[theorem]{Fact}

\newcommand{\AutoAdjust}[3]{\mathchoice{ \left #1 #2  \right #3}{#1 #2 #3}{#1 #2 #3}{#1 #2 #3} }
\newcommand{\Xcomment}[1]{{}}

\newcommand{\InBrackets}[1]{\AutoAdjust{[}{#1}{]}}
\newcommand{\Ex}[2][]{\operatorname{\mathbf E}_{#1}\InBrackets{#2}}
\newcommand{\Exlong}[2][]{\operatornamewithlimits{\mathbf E}\limits_{#1}\InBrackets{#2}}
\newcommand{\Prx}[2][]{\operatorname{\mathbf{Pr}}_{#1}\InBrackets{#2}}

\newcommand{\ind}[1]{\mathds{1}\left[\vphantom{\sum}#1\right]}

\newcommand{\piv}[1][v]{p_i^{#1}}
\newcommand{\pjv}{p_j^{v}}
\newcommand{\aiuv}{a_{i}^{(u,v)}}
\newcommand{\Aiu}[1][u]{A_i^{#1}}
\newcommand{\Amiu}[1][u]{A_{\text{-}i}^{#1}}
\newcommand{\giuv}{g_{i}^{(u,v)}}
\newcommand{\xiuv}{x_{i}^{(u,v)}}
\newcommand{\xiu}{x_i^u}
\newcommand{\xu}{x^u}
\newcommand{\hu}{h_u}

\newcommand{\xju}{x_j^u}

\newcommand{\ziuv}{(2 \rhoiuv - 1)^+}
\newcommand{\rhoiuv}[1][u]{\rho_{i}^{(#1,v)}}
\newcommand{\rhojuv}{\rho_{j}^{(u,v)}}

\newcommand{\muiv}[2][u]{\mu_i^v(#1 \to #2)}

\newcommand{\eqdef}{\overset{\mathrm{def}}{=\mathrel{\mkern-3mu}=}}
\newcommand{\dd}[1]{\;\mathrm{d}#1}

\newcommand{\ratmsm}{\Gamma_{\text{MAM}}}
\newcommand{\ratrcrs}{\Gamma_{\text{CAR}}}

\DeclareMathOperator*{\argmax}{argmax}

\newcommand{\threshold}{\beta}

\newcommand{\dif}{\mathrm{d}}

\title{\Large Prophet Secretary and Matching: the Significance of the Largest Item}

\author{
Ziyun Chen\thanks{University of Washington. 
Email: \texttt{ziyuncc@cs.washington.edu}. This work was done when the author was at Tsinghua University.}
\and 
Zhiyi Huang\thanks{The University of Hong Kong. 
Email: \texttt{zhiyi@cs.hku.hk}}
\and 
Dongchen Li\thanks{The University of Hong Kong. 
Email: \texttt{dongchen.li@connect.hku.hk}}
\and 
Zhihao Gavin Tang\thanks{ITCS, Key Laboratory of Interdisciplinary Research of Computation and Economics, Shanghai Institute of International Finance and Economics, Shanghai University of Finance and Economics. Email: \texttt{tang.zhihao@mail.shufe.edu.cn}}
\and 
}
\date{}

\begin{document}

\maketitle

\begin{abstract} 
    The prophet secretary problem is a combination of the prophet inequality and the secretary problem, where elements are drawn from known independent distributions and arrive in uniformly random order. In this work, we design 1) a $0.688$-competitive algorithm, that breaks the $0.675$ barrier of blind strategies (Correa, Saona, Ziliotto, 2021), and 2) a $0.641$-competitive algorithm for the prophet secretary matching problem, that breaks the $1-1/e\approx 0.632$ barrier for the first time. Our second result also applies to the query-commit model of weighted stochastic matching and improves the state-of-the-art ratio (Derakhshan and Farhadi, 2023).
\end{abstract}

\section{Introduction}
\label{sec:intro}
The study of prophet inequality dates back to the 1970s~\cite{krengel1977semiamarts,krengel1978} from optimal stopping theory. Consider $n$ items with independent random values arriving one by one in an adversarial order. The value distribution $F_i$ of each item $i$ is known upfront to the algorithm, but the realization of value $v_i \sim F_i$ is only revealed on the item's arrival. After seeing the item's identity $i$ and value $v_i$, the algorithm decides immediately whether to accept the item and collect its value; the algorithm can accept at most one item in this problem. The goal is to maximize the expected value of the accepted item and compete against the prophet, i.e., the expected maximum value $\Ex{\max_i v_i}$. It is known that the optimal competitive ratio is $\frac{1}{2}$ for this problem.

A fundamental extension of the prophet inequality is prophet matching. Consider an underlying bipartite graph with edge weights drawn from known distributions. The vertices on one side are known upfront and those on the other side arrive online. On the arrival of an online vertex $v_i$, the weights of its incident edges are revealed and the algorithm decides whether to match $v_i$ and to which offline vertex. The classic prophet inequality is captured by this model with one offline vertex. \citet*{soda/FeldmanGL15} gave a tight $\frac{1}{2}$ competitive algorithm for the matching setting, and their result was further generalized to settings when all vertices arrive online~\cite{mor/EzraFGT22}.

In this work, we consider the secretary variants (a.k.a.\ the random order variants) of prophet inequality and prophet matching, i.e., the setting where the arrival order of items (resp.\ vertices) is uniformly at random.

The study of prophet secretary was initiated by \citet{esa/EsfandiariHLM15}, who designed a $1-1/e \approx 0.632$ competitive algorithm and provided an upper bound of $0.75$.%
\footnote{The competitive ratio of an algorithm is a number between $[0,1]$. A lower bound corresponds to an algorithm and an upper bound corresponds to an impossibility result.} 
Since then, a sequence of follow-up works~\cite{ec/AzarCK18, mp/CorreaSZ21, corr/Harb23, ec/BubnaC23, corr/Giambartolomei23} have focused on closing the gap. The state-of-the-art lower and upper bounds are $0.672$ by \citet{corr/Harb23} and $0.723$ by \citet{corr/Giambartolomei23} respectively.

Less progress has been made on the prophet secretary matching problem.
\citet{soda/EhsaniHKS18} gave a $1-1/e$ competitive algorithm.
Very recently, the $1-1/e$ barrier was surpassed in two special cases: 1) the i.i.d.\ setting studied by \citet{soda/Yan24} and \citet{wine/QiuFZW23}; and 2) the query-commit setting studied by \citet{soda/Derakhshan023}. 
Beating $1-1/e$ for the general case of prophet secretary matching remains one of the most intriguing open questions to the online algorithms community.

\subsection{Our Contributions}


\paragraph{Result for Prophet Secretary.}
We design a $0.688$ competitive algorithm for the prophet secretary problem. Besides the improvement over the state-of-the-art $0.672$ ratio,
our result further surpasses the $0.675$ barrier of \emph{blind strategies}~\cite{mp/CorreaSZ21}, the family of algorithms that Correa et al.~\cite{mp/CorreaSZ21} and Harb~\cite{corr/Harb23} focused on. 
Blind strategies rely on only the distribution of the maximum value, but not the fine-grained distributional information of individual items' values.
Intuitively, such fine-grained information must be crucial because the items are heterogeneous in the prophet secretary problem.
However, it is technically challenging to incorporate such information to design and analyze item-dependent strategies: 
changing the strategy for one item would unavoidably affect the probability of accepting other items since we can accept only one of them.

\paragraph{Technique: Activation-Based Algorithms.}
We introduce two ideas to address this difficulty.
First, we change our point of view from designing acceptance probabilities to choosing \emph{activation rates}.
In general, an online algorithm is defined by the probability of accepting an item based on its identity $i$, value $v$, and the set of future items that will arrive later.
Following the conventional wisdom, the current item's arrival time $t$ is a good surrogate for aggregating information about exponentially many possible sets of future items over their random arrivals.
Here, we interpret the random order as having each item arrive within a time horizon from $0$ to $1$ uniformly at random.
In short, algorithms are represented by the acceptance probabilities for item $i$ when it has value $v$ and arrives at time $t$.
However, it is difficult to analyze the algorithms based on this representation.

By contrast, we will consider the activation rates $a_i^v(t)$ for item $i$ when it has value $v$ and arrives at time $t$, and the overall activation rates $A_i(t) = \Ex[v \sim F_i]{a_i^v(t)}$ of the item at time $t$. 
We activate this item (and accept it if no item has been accepted yet) with probability:
\[
	g_i^v(t) = a_i^v(t) \cdot e^{- \int_0^t A_i(x) \dif x}
	~.
\]

This new viewpoint offers two useful invariants.
By definition, the probability that we activate item $i$ before time $t$ equals:
\[
	\int_0^t \Ex[v \sim F_i]{g_i^v(x)} \dif x ~=~ 1 - e^{- \int_0^t A_i(x) \dif x}
	~.
\]

Hence, the activation events effectively follow a Poisson process with rates $A_i(t)$.
Accordingly, the probability that we activate item $i$ with value $v$ and arrival time $t$ is:
\[
	\Prx{v_i = v} \cdot a_i^v(t) \cdot \underbrace{\vphantom{\big|} e^{- \int_0^t \sum_{j=1}^n A_j(x) \dif x}}_{(\star)} \dif t
	~.
\]

Note that the second part $(\star)$ is independent of the item's identity $i$ and value $v$.
Therefore, we can simplify the dependence of different items' strategies by introducing an upper bound on $\sum_{j=1}^n A_j(t)$ for any time $t$.
Subject to this invariant, we can freely design the activation rates $a_i^v(t)$ for each item $i$ and value $v$ to approximately match its contribution to the prophet benchmark.

To further simplify the analysis, we focus on activation rates $a_i^v(t)$ that are step functions that change their values at a common threshold time.
We demonstrate the effectiveness of this viewpoint and such simple step activation rates in \Cref{sec:prophet-secretary-step-activation} by proving a $0.694$ competitive ratio when all items are small in the sense that each contributes only $o(1)$ to the prophet benchmark.

\paragraph{Technique: Significance of the Largest Item.}
To further handle the general case of prophet secretary, we will focus on the item $i_0$ with the largest probability of being selected by the prophet.
This is partly inspired by the existing hard instances (e.g., \cite{mp/CorreaSZ21,ec/BubnaC23,corr/Giambartolomei23}), all of which involve one large item and many small items.
The significance of the largest item is twofold:
1) we need to design a special strategy for it beyond the step-function activation rates, and 2) its characteristics provide sufficient information for selecting the invariants for the other items' activation rates.

Why do we need a special strategy for this largest item?
Consider the extreme case when it is the only item that matters.
Intuitively, we would like to select it with certainty on its arrival.
However, we cannot do that using the step-function activation rates.
Within a time interval where the activation rates $a_i^v(t)$ remain a constant, the $e^{- \int_0^t A_i(x) \dif x}$ term decreases the activation probability over time.
This decrease is mild for smaller items, but could be substantial for the largest item.
Remarkably, this seemingly trivial instance plays an important role of establishing the $0.675$ barrier for blind strategies~\cite{mp/CorreaSZ21}.
Motivated by this extreme case, we let the largest item's activation probability rather than its activation rate be piece-wise constant.
While it is difficult to analyze algorithms based on the representation by activation/acceptance probabilities in general, we show that it is manageable to do that for just one largest item. 
This special treatment of just one largest item is sufficient for breaking the $0.675$ barrier.

We will consider two characteristics of this largest item:
its probability of being selected by the prophet, i.e., $x_0 = \Prx{v_{i_0} = \max_j v_j}$, and a quantity $h_0$ that measures the extent to which item $i_0$ would be selected by the prophet with probability more than half, over the randomness of the other items' values.
Based on just $x_0$ and $h_0$, we will choose the (1) invariants $\sum_{j \ne i_0} A_j(t)$ for the activation rates of other items, (2) the shared threshold time for the step-function activation rates of other items, and (3) the three-stage step-function activation probabilities of the largest item $i_0$.
It is surprising that these two characteristics of the largest item alone are sufficient for choosing all important invariants for our algorithm and analysis.





\paragraph{Result for Prophet Secretary Matching.}

We design a $0.641$-competitive algorithm for the prophet secretary matching problem, breaking the $1-1/e$ barrier for the first time.
As a corollary of this result, we also improve the state-of-the-art ratio of the query-commit setting from $0.633$~\cite{soda/Derakhshan023} to $0.641$ through a reduction~\cite{soda/GamlathKS19, icalp/CostelloTT12} from the query-commit model to the secretary model.


\paragraph{Summary of Techniques for Prophet Secretary Matching.}
We start by extending the activation-based framework to matching.
First, let us consider a simple strategy:
upon the arrival of an online vertex, assign it to an offline vertex with probability proportional to how likely the prophet would match them.
We remark that the assignment is \emph{independent of the arrival of earlier vertices and their matching results}. 
Then, from each offline vertex's viewpoint, it may treat the online vertices (more precisely, the corresponding edges) as online items in the prophet secretary problem, treating those not assigned to it as having zero values.

However, the online stochastic matching literature suggests that we should not na\"ively follow this approach and apply the two-stage step-function activation rates from prophet secretary, or we would miss the opportunity of exploiting second-chance (re-)assignments (a.k.a.\ the power of two choices).
If an offline vertex is already matched, we should no longer assign online vertices to it, but instead redirect the opportunities to other unmatched offline vertices.
Hence, we introduce a third stage into the activation-based algorithm, which has the same activation rates as the second stage, but takes into account the assignments redirected from the other offline vertices.
This may be viewed as reinterpreting the three-stage algorithm by \citet{soda/Yan24} for the i.i.d.\ special case within the activation-based framework.
By doing so, we achieve the same $0.645$ competitive ratio but more generally for all non-i.i.d.\ instances in which all edges are small, i.e., when each edge contributes only $o(1)$ to the prophet benchmark.

On the other hand, the worst-case competitive ratio of this approach degenerates to $1-1/e$ if there is a large edge adjacent to every offline vertex.
To complement this scenario, we introduce a variant of the random order contention resolution scheme (RCRS) algorithm for matching~\cite{esa/LeeS18,icalp/FuTWWZ21}.
This may be viewed under the activation-based framework as follows.
For each offline vertex $u$ and its largest edge $(u,v)$, let its adjacent edges other than $(u,v)$ have constant activation rates;
let edge $(u, v)$'s activation probability be a $0$-$1$ step-function.
This is consistent with our approach for the prophet secretary problem, but the design of activation rates and probabilities is simpler due to the complications in the analysis of the more general matching problem, and the fact that we only need to beat the $1-\frac{1}{e}$ barrier in this case.

We show that a hybrid algorithm that randomizes over the above two approaches achieves the stated $0.641$ competitive ratio.

\subsection{Related Works}
\citet{duttingProphetSecretaryOnline2023} studied the computational complexity of the optimal online algorithm for prophet secretary and gave a PTAS, though it does not imply any competitive ratio of the optimal online algorithm. Abolhassani et al.~\cite{stoc/AbolhassaniEEHK17} and Liu et al.~\cite{ec/LiuLPSS21} studied the prophet secretary problem under small-item assumptions. They proved that if either 1) every distribution appears sufficiently many times~\cite{stoc/AbolhassaniEEHK17, ec/LiuLPSS21} or 2) every distribution has only a negligible probability of being non-zero~\cite{ec/LiuLPSS21}, there exists a $0.745$-competitive algorithm, matching the optimal competitive ratio as in the i.i.d.\ setting.

The order-selection prophet inequality lies between the i.i.d.\ setting and the secretary setting. In this variant, the algorithm is given the extra power of selecting the arrival order of the items. This is motivated by the application of prophet inequalities to sequential posted pricing mechanisms, and has been studied by~\cite{stoc/ChawlaHMS10,or/BeyhaghiGLPS21,focs/PengT22,ec/BubnaC23}. The current state-of-the-art competitive ratio is $0.725$ by Bubna and Chiplunkar~\cite{ec/BubnaC23}.

Besides matching, the prophet inequality has also been generalized to other combinatorial settings, including matroids~\cite{geb/KleinbergW19}, combinatorial auctions~\cite{soda/FeldmanGL15, siamcomp/DuttingFKL20}, and general downward-closed constraints~\cite{stoc/Rubinstein16}. Ehsani et al.~\cite{soda/EhsaniHKS18} also studied the prophet secretary problem under matroid constraints and achieved a competitive ratio of $1-1/e$. Beating this ratio for general matroid constraints remains an important open question.

Finally, the unweighted and vertex-weighted online stochastic matching problems have attracted a lot of attention in the online algorithms community~\cite{focs/FeldmanMMM09,esa/BahmaniK10,wine/HaeuplerMZ11,mor/ManshadiGS12,mor/JailletL14,algorithmica/BrubachSSX20a,DBLP:conf/stoc/HuangS21,stoc/TangWW22}. Most of these works assumed i.i.d. arrivals of online vertices, in order to surpass the optimal $1-1/e$ competitive ratio of online (vertex-weighted) bipartite matching~\cite{stoc/KarpVV90, soda/AggarwalGKM11}.


\section{Prophet Secretary}


\subsection{Model}

Throughout the paper, we use the continuous arrival time formulation of the prophet secretary problem. 
There are $n$ items with non-negative values $v_1,\dots, v_n$ drawn independently from distributions $F_1,\dots, F_n$.
For the ease of exposition, we make two assumptions without loss of generality.
First, we consider \emph{discrete distributions}, since continuous distributions can be approximated up to an arbitrary accuracy by discrete distributions.
Second, we assume that the distributions have \emph{disjoint supports} so that the maximum value is always unique. 
This can be achieved without loss of generality by adding an infinitesimal $\varepsilon_i$ to the values from $F_i$ for tie-breaking.
We shall use $p_i^v$ to denote the probability that $v_i =v$.

The items arrive within a time horizon between $0$ and $1$, where the arrival time $t_i$ of item $i$ is drawn i.i.d.\ from the uniform distribution over $[0,1]$.
The algorithm, only knowing $F_1,\dots ,F_n$, is presented with item $i$ at its arrival time $t_i$. 
Upon item $i$'s arrival, the algorithm sees the its identity $i$, realized value $v_i$, and arrival time $t_i$,
and must decide whether to accept or reject item $i$ immediately.
If the algorithm accepts item $i$, it collects the value $v_i$; otherwise, the algorithm continues to the next item. 
The algorithm can accept at most one item and the goal is to maximize the expected value of the accepted item.

Following standard competitive analysis, we compare an online algorithm's expected objective to the expected offline optimum (a.k.a., the prophet), i.e. $\Ex{\max_i v_i}$.
An algorithm is $\Gamma$-competitive if its expected accepted value is at least $\Gamma$ times the expected offline optimum.

\paragraph{Continuous Arrival Time vs.\ Discrete Arrival Order.} We remark that the continuous time model is equivalent to the classical discrete random order model of prophet secretary.
The discrete random order model samples a permutation $\pi$ of the items uniformly at random over all $n!$ possible permutations.
Then, the items arrive in a sequence according to $\pi$. 
On one hand, any algorithm for the discrete random order model can be simulated in the continuous time model, by letting $\pi$ be the permutation of the items in ascending order of their arrival times $t_i$.
On the other hand, any algorithm for the continuous time model can be simulated in the discrete random order model by first sampling $t_{(1)} \leq t_{(2)} \leq \cdots \leq t_{(n)}$ i.i.d.\ from the uniform distribution over $[0,1]$, sorted in ascending order, and letting $t_{(\pi(i))}$ be the arrival time of item $i$.



\subsection{Our Result}

\begin{theorem}
\label[theorem]{thm:prophet-secretary}
    There exists a $0.688$-competitive algorithm for prophet secretary.
\end{theorem}

Prior to our work, the state-of-the-art is a $0.672$-competitive algorithm~\cite{correaProphetSecretaryBlind2021, harbFishingBetterConstants2023}. This algorithm is a \emph{blind strategy};
\citet*{correaProphetSecretaryBlind2021} showed that no blind strategy can be better than $0.675$ competitive.
Our result breaks this barrier. 

Our result can also be viewed as an exploration on the trade-off between the complexity of the state space and the approximate optimality of the online algorithms.
On one hand, the online optimal algorithm based on dynamic programming has an exponential-size state space.
When it decides whether to accept an item $i$, the state consists of the item's identity $i$, value $v_i$, arrival time $t_i$, and the subset of items that have not yet arrived.
In particular, the last component makes the state space exponentially large (see \citet{duttingProphetSecretaryOnline2023} for a PTAS that discretizes of the state space down to polynomial in $n$ but doubly exponential in $\frac{1}{\varepsilon}$).
On the other hand, all existing algorithms in the literature of competitive analysis for prophet secretary, including the aforementioned blind strategies, decide whether to accept an item based only on this item's value and arrival time, but oblivious to its identity and the subset of remaining items.
By contrast, our algorithm will further take into account the current item's identity, on top of its value and arrival time, to decide whether to accept the item.





The rest of the section is organized as follows.
\Cref{subsec:prophet-framework} introduces the framework of activation-based online algorithms, and how to analyze their competitive ratios.
Then, \Cref{subsec:warm-ups} demonstrates two simple applications of this framework as warm-ups.
Finally, \Cref{subsec:step-act-rate} presents a $0.688$-competitive algorithm under this framework, and proves \Cref{thm:prophet-secretary}.

\subsection{Activation-Based Online Algorithms}
\label{subsec:prophet-framework}

\subsubsection{General Framework}

An \emph{activation-based online algorithm} is as follows. 
Upon the arrival of an item, we see its identity $i$, value $v \sim F_i$, and arrival time $t$.
Note that we suppress the subscript $i$ in the value and arrival time for notation simplicity as the item's identity $i$ will always be clear from context.
Then, we decide whether to activate the item or not based on an \emph{activation probability} $0 \leq g_i^{v}(t) \leq 1$, depending on $i$, $v$, and $t$, but importantly, independent of the subset of items that have not yet arrived.
Finally, we accept an item if it is the first activated item.
By definition, an activation-based online algorithm is characterized by the activation probabilities $g_i^v(t) \in [0,1]$.

\medskip
\begin{tcolorbox}[title=Activation Framework]
For time $t$ from $0$ to $1$:
\begin{enumerate}[(1)]
    \item If item $i$ arrives with value $v$, activate it with activation probability $g_i^v(t)$.
    \item Accept item $i$ if it is the first activated item.
\end{enumerate}
\end{tcolorbox}
\medskip


Now we show two basic properties of these algorithms.

\begin{lemma}
	\label{lem:activation-probability}
	The probability of activating item $i$ before time $\theta$ is $\int_0^\theta \Ex[v \sim F_i]{g_i^{v}(t)} \dd t$.
\end{lemma}

\begin{proof}
	We will first consider the probability of activating item $i$ \emph{with value $v$} before time $\theta$. 
	This event can be decomposed into three sub-events:
	(1) the item arrives at time $t \leq \theta$;
	(2) the item has value $v$, and 
	(3) the algorithm activates the item.
	
	Recall that the arrival time $t$ distributes uniformly over $[0, 1]$, and the probability of drawing value $v$ from $F_i$ is $\piv$.
	Note that $t$ and $v$ are independent.
	Further, conditioned on the first two sub-events, the third happens with probability $g_i^v(x)$. 
	Hence, the probability of activating item $i$ with value $v$ before time $\theta$ is:
	\[
	    \int_0^\theta  p_i^{v} g_i^{v}(t)\dd t
		~.
	\]	
	
	Summing over all possible values $v$, the probability of activating a item $i$ before time $\theta$ is: $\int_0^\theta \sum_v p_i^{v} g_i^{v}(t) \dd t = \int_0^\theta \Ex[v \sim F_i]{g_i^{v}(t)}\dd t$.
\end{proof}

\begin{lemma}
\label{lem:acceptance-probability}
	The probability of accepting item $i$ with value $v$ before time $\theta$ is:
	\[
		\int_0^\theta p_i^v g_i^{v}(t) \prod_{j\neq i} \left( 1- \int_0^{t} \Exlong[v \sim F_j]{g_j^{v}(x)} \dd x\right) \dd t
		~.
	\]
\end{lemma}

\begin{proof}
For any time $t\leq \theta$, the event that item $i$ is accepted at time $t$ can be decomposed into three independent sub-events:
(1) item $i$ is has value $v$ at time $t$; (2) item $i$ is activated at time $t$, and
(3) any item $j\neq i$ is not activated before time $t$. The first sub-event happens with probability $p_i^v$, and the second sub-event happens with probability $g_i^v(t)$. By \Cref{lem:activation-probability}, the probability of activating item $j$ before time $t$ is $\int_0^{t} \Ex[v \sim F_j]{g_j^{v}(x)} \dd x$, thus the third sub-event happens with probability $\prod_{j\neq i} \big( 1- \int_0^{t} \Ex[v \sim F_j]{g_j^{v}(x)} \dd x \big)$. 
Together, the probability of accepting item $i$ with value $v$ before time $\theta$ is as stated.
\end{proof}
\subsubsection{Activation Rates}
\label{subsec:act-rate-reconstruction}

This subsection presents an alternative representation of the activation probabilities $g_i^v(t)$ that will simplify parts of the subsequent analysis.
Formally, for each item $i$, we focus on activation probabilities in the following form:
\[
	g_i^{v}(t) = a_i^v(t)e^{-\int_0^t A_i(x) \dd x}~, \quad \forall i,v,t
\]
where $A_i(x)\eqdef \sum_{v} p_i^v a_i^v(x)=\Ex[v \sim F_i]{a_i^v(x)}$.
We call $a_i^v(t)$ the \emph{activation rate} of $v$ at time $t$.
We say that the activation rate $a_i^v(t)$ is \emph{well-defined} if the corresponding activation probability $g_i^v(t)$ is at most $1$ for every $i,t,v$.

\begin{lemma}
\label{lemma:single-act-rate}
The probability of activating item $i$ before a threshold time $\theta$ is $1-e^{-\int_0^\theta A_i(t) \dd t}$.
\end{lemma}

This lemma justifies why we call $a_i^v(t)$ the activation rate, as its conclusion resembles the form of a Poisson activation process with rate $A_i(t)$.
\begin{proof}
Replacing $g_i^v (t)$ with $a_i^v(t)e^{-\int_0^t A_i(x)\dd x}$ in \Cref{lem:activation-probability}, the probability of activating item $i$ before time $\theta$ is:
\[
\int_0^\theta \Ex[v]{a_i^{v}(t)} e^{-\int_0^t A_i(x) \dd x} \dd t = \int_0^\theta A_i(t) e^{-\int_0^t A_i(x) \dd x} \dd t = 1 - e^{-\int_0^\theta A_i(t) \dd t}~.
\]
\end{proof}

\begin{lemma}
\label{lem:all-act-rate}
	The probability of accepting item $i$ with value $v$ is:
	\[
        \int_0^1 p_i^v a_i^v(t) e^{-\int_0^t \sum_j A_j(x) \dd x} \dd t
        ~.
	\]
\end{lemma}

\begin{proof}
We substitute $g_i^v (t)$ with $a_i^v(t)e^{-\int_0^t A_i(x)\dd x}$ in \Cref{lem:acceptance-probability} and the probability that item $j$ is not activated before time $t$ with $e^{-\int_0^t A_j(x) \dd x}$ by \Cref{lem:all-act-rate}. Hence, the probability of accepting item $i$ with value $v$ is:
\[
\int_0^1 p_i^v a_i^v(t) e^{-\int_0^t A_i(x) \dd x} \prod_{j\neq i} \left( e^{-\int_0^t A_j(x) \dd x} \right) \dd t
= \int_0^1 p_i^v a_i^v(t) e^{-\int_0^t \sum_j A_j(x) \dd x} \dd t~.
\]
\end{proof}

\subsubsection{Sufficient Condition for Competitive Analysis}
\label{subsec:offline-opt-competitive-ratio}
For every item $i$ and value $v$, we denote the (conditional) probability that $v_i=v$ is the maximum value as the following:
\[
x_i^v \eqdef p_i^v \cdot \prod_{j\neq i} \Prx{v_j < v} \quad \text{and} \quad \rho_i^v \eqdef \prod_{j \neq i} \Prx{v_j < v}~.
\]
%

We list the following properties of $x_i^v$ and $\rho_i^v$ that are straightforward to verify.
\begin{lemma}
	\label{lem:prophet-secretary-properties}
1) $x_i^v = p_i^v \rho_i^v$; 2) $\rho_i^v\leq 1$ is non-decreasing in $v$; and 3) $\sum_{i,v} x_i^v = 1$.
\end{lemma}

Our competitive analysis is induced by the following stochastic dominance.
\begin{lemma}
	\label{lem:prophet-secretary-sufficient-condition}	
	An online algorithm is $\Gamma$-competitive if for any item $i$ and any value $v$:
	\[
	    \int_0^1 p_i^v g_i^{v}(t) \prod_{j\neq i} \left( 1- \int_0^{t} \Exlong[v_j \sim F_j]{g_j^{v_j}(x)} \dd x\right) \dd t \geq \Gamma \cdot x_i^v = \Gamma \cdot p_i^v \rho_i^v
	    ~.
	\]
\end{lemma}
\begin{proof}
    Multiplying the inequality by $v$ and summing over all $v$ and $i$, the left hand side gives the gain from the algorithm, and the right hand side is $\Gamma$ of the offline optimal gain.
\end{proof}

%
%
%

\subsection{Warm-ups}
\label{subsec:warm-ups}

This subsection presents two warm-up algorithms that use simple activation rates.
The first algorithm considers a constant activation rate $a_i^v(t) = \rho^v_i$ for any item $i$ and any value $v$.
The second algorithm considers a threshold time $\threshold$, and lets $a_i^v(t)$ be a step function that changes its value at time $\threshold$ for any item $i$ and value $v$.

\subsubsection{Constant Activation Rates}
\label{subsubsec:constant-act-rate}


We now present the Constant Activation Rates algorithm, and prove that it is $1-\frac{1}{e}$ competitive.
We note that it is exactly the $(1-\frac{1}{e})$-selectable RCRS algorithm by \citet{esa/LeeS18}, rewritten in our activation-based framework.

\smallskip

\begin{tcolorbox}[title=Constant Activation Rates~\cite{esa/LeeS18}]
    \label{algo:const-act-rate-0.632}
    For time $t$ from $0$ to $1$:
    \begin{enumerate}[(a)]
        \item If item $i$ arrives with value $v$, activate it with activation rate $a_i^v(t) = \rho_i^v$.
        \item Accept item $i$ if it is the first activated item.
    \end{enumerate}
\end{tcolorbox}

\begin{theorem}
    Constant Activation Rates is $\left(1-\frac{1}{e}\right)$-competitive.
\end{theorem}

\begin{proof}
	Recall that by definition of $A_i(t)$, we have $A_i(t) = \sum_v p_i^v \rho_i^v = \sum_v x_i^v$. According to \Cref{lem:all-act-rate}, the probability of accepting item $i$ with value $v$ is:
    \begin{align*}
    \int_0^1 p_i^v \rho_i^v e^{-\sum_{j}\int_0^t A_{j}(x)\dd x} \dd t
    ~=~ &x_i^v \int_0^1 e^{-t \cdot \sum_{j}\sum_v x_j^v} \dd t 
    ~=~ x_i^v \int_0^1 e^{-t} \dd t = \left(1-\frac{1}{e}\right) x_i^v.
    \end{align*}
    Here we use the property that $\sum_j \sum_v x_j^v = 1$ (\Cref{lem:prophet-secretary-properties}).
    Further by \Cref{lem:prophet-secretary-sufficient-condition}, we get that the algorithm is $\left(1-\frac{1}{e}\right)$-competitive .
\end{proof}

Why are constant activation rates suboptimal?
Intuitively, we would like to let the activation rate be larger, e.g., converging to $1$ when the time approaches $1$, because at that time there would be few items left and thus, we should take any occurring item. 
Accordingly, the activation rate should be smaller at the beginning, unless the realized value is very high, because we anticipate more items to come in the future.


\subsubsection{Step Activation Rates}
\label{sec:prophet-secretary-step-activation}

The next warm-up algorithm implements the above idea by letting the activation rate $a_i^v(t)$ for any item $i$ and value $v$ be a step function that increases at a threshold time $\threshold$.
We maintain the invariant that the activation rates in the two stages sum to $2 \cdot \rho^v_i$, but let the rate in the second stage be as high as possible.


\begin{tcolorbox}[title=Step Activation Rates]
    \label{algo:step-act-rate-0.694}
    Parameter: $\threshold = 0.367$.\\[2ex]
	For time $t$ from $0$ to $1$:
    \begin{enumerate}[(a)]
        \item If item $i$ arrives with value $v$, activate it with activation rate:
        \[
        a_i^v(t)=\left\{ \begin{aligned}
            (2\rho_i^v - 1)^+ & \quad \text{if } t\in [0, \threshold),\\
            2\rho_i^v - (2\rho_i^v - 1)^+ & \quad \text{if } t\in [\threshold, 1].
        \end{aligned} \right.
        \]
        \item Accept item $i$ if it is the first activated item.
    \end{enumerate}
\end{tcolorbox}
    
\bigskip
    
We will analyze this algorithm under the assumption that:
\begin{equation}
	\label{eqn:jl-assumption}	
    h \eqdef \sum_{i,v} (2x_i^v - p_i^v)^+ \le 1 - \ln 2
    ~.
\end{equation}
This is inspired by a property from the stochastic online matching literature proved by \citet{mor/ManshadiGS12}.
It would hold up to an $o(1)$ additive error if every item $i$ is small in the sense that it has the highest value with probability only $o(1)$, namely, $\sum_v x_i^v=o(1)$. A formal proof of this claim is deferred to \Cref{app:h-function}.

\paragraph{Parameter of the Invariant.}
We remark that there is actually a parameter $s$ for the invariant that we can optimize for the above algorithm: the activation rates in the two stages sum to $s \cdot \rho_i^v$. In the main body of the paper, we fix $s$ to be $2$ both in this warm-up algorithm and our main algorithm in the next subsection. This simplification makes it clear the connection between the assumption \eqref{eqn:jl-assumption} and the constraint~\cite{mor/JailletL14} from previous works. The algorithm for general $s$ is provided in \Cref{app:algo-general-s}.

\begin{lemma}
\label{cl:ai}
    We have: 
    \[
    	\sum_j A_j(t) =
    	\begin{cases}
    		h & \mbox{if $0 \le t < \threshold$;} \\
    		2-h & \mbox{if $\threshold \le t \le 1$.}
    	\end{cases}
	\]
\end{lemma}

\begin{proof}
    By the definition of $A_j(t)$'s, we have 
    $\sum_j A_j(t) = \sum_j \sum_v p_j^v a_j^v(t)$. Further by the definition of $h$, we have $\sum_j \sum_v p_j^v (2\rho_j^v-1)^+=\sum_j \sum_v (2x_j^v - p_j^v)^+=h$. Then, according to \Cref{lem:prophet-secretary-properties}, we have $\sum_j \sum_v p_j^v \rho_j^v =\sum_j \sum_v x_j^v =1$, and thus, $\sum_j \sum_v p_j^v (2\rho_j^v-(2\rho_j^v-1)^+)=2-h$.
\end{proof}


\begin{theorem}
	\label{thm:step-activation}
	Assuming~\eqref{eqn:jl-assumption}, Step Activation Rates is $0.694$-competitive.	
\end{theorem}

\begin{proof}
We denote the cumulative total activation rates up to time $t$ as:
\[
 	G(t) \eqdef \int_0^t \sum_j A_j(t) \dd y =
 	\begin{cases}
		h \cdot t & \mbox{if $0 \le t < \threshold$;} \\
		h \cdot \threshold + (2-h) \cdot (t-\threshold) & \mbox{if $\threshold \le t \le 1$.}
 	\end{cases}
\]

By \Cref{lem:all-act-rate}, the probability of accepting item $i$ when its value is $v$ equals:

\begin{align*}
    & \int_0^1 \piv a_i^v(t) e^{-\int_0^t \sum_j A_j(t) \dd y} \dd t \\
    & \qquad
    = ~ \piv \cdot \bigg((2\rho_i^v - 1)^+ \int_0^\threshold e^{-G(t)} \dd t + \big( 2 \rho_i^v - (2\rho_i^v - 1)^+ \big) \int_{\threshold}^1 e^{-G(t)} \dd t \bigg)\\
    & \qquad
    = ~ \piv \cdot \bigg( (2\rho_i^v - 1)^+ \cdot \underbrace{\frac{1 - e^{-h \threshold}}{h}}_{(a)} ~+~ \big( 2 \rho_i^v - (2\rho_i^v - 1)^+ \big) \cdot \underbrace{\frac{e^{-h \threshold} - e^{-(2-h) + (2-2h) \threshold}}{2-h}}_{(b)} \bigg)
    ~.
\end{align*}

Our choice of threshold time $\beta \approx 0.367$ makes $(a) \approx (b) > 0.347$, i.e., half the claimed competitive ratio, when $h = 1 - \ln 2$.
For smaller $h$, we still have $(a), (b) > 0.347$ because both $(a)$ and $(b)$ are decreasing in $h$.
%
%
%
%
Hence the above bound is at least:
\[
	\piv \cdot \Big( (2\rho_i^v - 1)^+ + \big( 2 \rho_i^v - (2\rho_i^v - 1)^+ \big) \Big) \cdot 0.347 = 0.694 \cdot \piv \rho_i^v
	~.
\]

By~\Cref{lem:prophet-secretary-sufficient-condition}, this algorithm is $0.694$-competitive.
\end{proof}

\subsection{Step Activation Rates Except One Large Item}
\label{subsec:step-act-rate}


Finally, we present our algorithm for the general case of prophet secretary.
As mentioned in the last subsection, we shall optimize the parameter $s$ from the invariant to maximize our competitive ratio. Within this subsection, we fix $s=2$ and establish a competitive ratio of $0.686$. The more general version of our algorithm and its analysis, which prove the $0.688$ competitive ratio stated in \Cref{thm:prophet-secretary}, are provided in \Cref{app:algo-general-s}.

Recall that the analysis in \Cref{sec:prophet-secretary-step-activation} relies on the assumption that $\sum_{i,v} (2x_i^v - p_i^v)^+ \le 1 - \ln 2$.
To get rid of the assumption, we consider the importance of the largest item $i_0=\arg \max_i \sum_v x_i^v$:
\[
x_0 \eqdef \sum_v x_{i_0}^v = \Prx{v_{i_0} = \max_{j} v_j}~.
\]
\begin{lemma}
\label{lem:hu_xu}
    For any instance of prophet secretary in which the largest item $i_0$'s probability of having the highest value is $x_0$, we have:
    \[
        \sum_{i,v} \big( 2 x_i^v - p_i^v \big)^+
        ~\le~
        h_2(x_0) ~\eqdef~ \max_{t\in [0,x_0)} ~ t + \sum_{k \ge 0} \Big( 2x_0 - \frac{x_0}{1 - t - kx_0} \Big)^+
        ~.
    \]
    Moreover, $\lim_{x\to 0^+} h_2(x) = 1-\ln 2$, and $h_2(1)=1$.
\end{lemma}

Intuitively, $(2 x_i^v - p_i^v)^+$ captures the extent to which an item-value pair $(i,v)$ is the largest more than half the time when it is present.
\Cref{lem:hu_xu} upper bounds the sum of this quantity over all item-value pairs, depending on the characteristics of the largest item.
It generalizes~\eqref{eqn:jl-assumption} by \citet{mor/ManshadiGS12} in the context of i.i.d.\ online stochastic matching, which plays an important role in that literature (e.g., \cite{mor/ManshadiGS12,mor/JailletL14,DBLP:conf/stoc/HuangS21,soda/Yan24}).

In \Cref{app:h-function}, we prove the following generalized statement to optimize over $s$.

\begin{lemma}
\label{lem:h_s_general}
For any instance of prophet secretary in which the largest item $i_0$'s probability of having the highest value is $x_0$, we have:
%
\[
    \frac{1}{s-1} \cdot \sum_{i,v} \big(sx_i^v - p_i^v \big)^+
    ~\le~
    h_s(x_0) 
    ~\eqdef~ \max_{t\in [0,x_0)} ~ t + \frac{1}{s-1}\sum_{k \ge 0} \Big(sx_0 - \frac{x_0}{1 - t - kx_0} \Big)^+
    ~.
\]
\end{lemma}


This immediately implies \Cref{lem:hu_xu}, which is the special case when $s = 2$. 
We will also use this lemma in the prophet secretary matching section.

\bigskip

\Cref{fig:prophet-secretary} presents our algorithm for the prophet secretary problem.
It handles the largest item and the other items using two different strategies.

For the largest item, we let its activation probabilities, rather than activation rates, be step functions that change values at some thresholds $\beta_0 < \beta_2$ to be determined.
We never activate this item from time $0$ to $\beta_0$.
The activation probabilities from time $\beta_0$ to $\beta_2$, and from time $\beta_2$ to $1$ sum to $2 \rho_{i_0}^v$, twice the probability that the largest item $i_0$ has the highest value \emph{conditioned on its value being $v$};
we further let the latter (from time $\beta_2$ to $1$) be as large as possible.

For the other items, we let their activation rates be step functions that change values at some threshold $\beta_1$ to be determined.
We will let $\beta_0 \le \beta_1 \le \beta_2$ and hence, the choice of their subscripts.
Like the warm-up case in the last subsection, we also let the activate rates before and after the threshold $\beta_1$ sum to $2 \rho_i^v$.
Unlike the warm-up case, where we let the activation rates before $\beta_1$ be as small as possible, here we let them be parameters $z_i^v$ that sum to (weighted by probability $p_i^v$) a fixed invariant depending on the characteristics of the largest item.


\bigskip

We next prove that the algorithm is well-defined.







\begin{figure}[ht]
\begin{tcolorbox}[title=Step Activation Rates Except One Large Item]
    \label{algo:step-act-rate-0.688}
    \begin{enumerate}[(1)]
        \item Let $i_0=\arg \max_i \sum_v x_i^v$ be the largest item, and let:
        	\[
        		x_0 = \sum_v x_{i_0}^v
        		\quad,\quad
        		z_{i_0}^v = (2 \rho_{i_0}^v - 1)^+
        		\quad,\quad
        		h_0 = \sum_{v} p_{i_0}^v z_{i_0}^v
        		~.
    		\]
        \item Calculate a set of $z_i^v$'s for the remaining items $i \ne i_0$ such that:
			\[
				h_{ot} \eqdef \sum_{i\neq i_0} \sum_v p_i^v z_i^v = \min \big\{h_2(x_0) - h_0, 1 - x_0 \big\}
				\quad,\quad
				z_i^v \in \big[ (2\rho_i^v - 1)^+, \rho_i^v \big]
				~,
			\]
			where $h_2(x_0)$ is the function as defined in \Cref{lem:hu_xu}.
        \item Select threshold times $\threshold_0 \leq \threshold_1 \leq \threshold_2$ based on $x_0$, $h_0$ to maximize \Cref{eqn:prophet-secretary-ratio}.
        \item For time $t$ from $0$ to $1$:
        \begin{enumerate}[(a)]
            \item If the largest item $i_0$ arrives with value $v$, activate it with probability:
            \[
	            g_{i_0}^v(t)=\left\{ \begin{aligned}
	                0 & \quad \text{if } t\in [0,\threshold_0),\\
	                z_{i_0}^v & \quad \text{if } t\in [\threshold_0, \threshold_2),\\
	                2\rho_{i_0}^v - z_{i_0}^v & \quad \text{if } t\in [\threshold_2, 1].
	            \end{aligned} \right.
            \]
            \item If an item $i\neq i_0$ arrives with value $v$, activate it with activation rate:
            \[
            a_i^v(t)=\left\{ \begin{aligned}
                z_i^v & \quad \text{if } t\in [0,\threshold_1),\\
                2\rho_i^v - z_i^v & \quad \text{if } t\in [\threshold_1, 1].
            \end{aligned} \right.
            \]
            \item Accept the item if it is the first activated item.
        \end{enumerate}
    \end{enumerate}
\end{tcolorbox}
    \caption{Algorithm for Prophet Secretary}
    \label{fig:prophet-secretary}
\end{figure}



\begin{lemma}
	There exist $z_i^v$'s satisfying the stated constraints in (2) of \Cref{fig:prophet-secretary}.
\end{lemma}

\begin{proof}
    Note that $ (2 \rho_{i}^v - 1)^+ \leq \rho_{i}^v$, it suffices to show that
    \[
        \sum_{i\neq i_0} \sum_v p_i^v (2 \rho_{i}^v - 1)^+ \leq \min\{h_2(x_0) - h_0, 1 - x_0\} \leq \sum_{i\neq i_0} \sum_v p_i^v \rho_i^v.
    \]
    
    The first inequality follows by:
    \[
        \begin{aligned}
            \sum_{i\neq i_0} \sum_v p_i^v (2 \rho_{i}^v - 1)^+
            &
            = \sum_{i\neq i_0} \sum_v (2x_i^v - p_i^v)^+ \\
            & 
            = \sum_{i} \sum_v (2x_i^v - p_i^v)^+ -h_0  \leq h_2(x_0) - h_0
            ~,
        \end{aligned}
    \]
    
    The second inequality follows by $\sum_{i\neq i_0} \sum_v p_i^v \rho_i^v = \sum_{i\neq i_0} \sum_v x_i^v = 1 - x_0$.
\end{proof}

\begin{lemma}
    The activation rates and probabilities are well-defined, namely they are between $0$ and $1$.
\end{lemma}

\begin{proof} 
    Recall that $0\leq x_i^v \leq \rho_{i}^v\leq 1$.
    For any item $i$, we have $0 \le z_{i}^v \leq \rho_{i}^v \leq 1$, and $2\rho_{i}^v - z_{i}^v = 2\rho_{i}^v - (2 \rho_{i}^v - 1)^+ = \min\{2\rho_{i}^v, 1\} \in [0,1]$.
    
    Hence, for the largest item $i_0$, the activation probability $g_{i_0}^v(t)$ is always between $0$ and $1$.
    For any other item $i \ne i_0$, the activation rate $a_i^v(t)$ is between $0$ and $1$ and the resulting activation probability $g_i^v(t) = a_i^v(t)e^{-\int_0^t A_i(x) \dd x}$ is also between $0$ and $1$.
\end{proof}


We are now ready to prove \Cref{thm:prophet-secretary}, which says that the Step Activation Rates Except One Large Item is $0.686$-competitive.

\begin{proof}[Proof of \Cref{thm:prophet-secretary}]
    We will leave $\threshold_0 \le \threshold_1 \le \threshold_2$ as parameters to be determined, and optimize them to maximize the competitive ratio via a computer-aided search.

    \paragraph{Largest Item.}
    By \Cref{lemma:single-act-rate}, the probability that an item $j \neq i_0$ was not activated before time $t$ is $e^{-\int_{0}^t A_j(x)\dd x}$, where $A_j(x)$ equals $\sum_v p_j^v z_j^v$ when $x\in [0, \threshold_1]$ and equals $2\sum_v x_j^v - \sum_v p_j^v z_j^v$ otherwise.

    By the definition of $h_{ot}$ and $\sum_j \sum_v x_j^v =1$ (\Cref{lem:prophet-secretary-properties}), we have:
    \[
        \sum_{j\neq i_0} A_j(x) = 
        \begin{cases} 
        \sum_{j \neq i_0} \sum_v p_j^v z_j^v = h_{ot}, & \text{for } x \in [0, \threshold_1] \\
        \sum_{j \neq i_0} (2\sum_v x_j^v - \sum_v p_j^v z_j^v) = 2(1 - x_0) - h_{ot}. & \text{for } x \in [\threshold_1, 1]
        \end{cases}
    \]
    
    Thus, the probability of not having any item $j\neq i_0$ activated before time threshold $t$ is:
    \[
        \prod_{j\neq i_0} e^{-\int_0^t A_j(x)\dd x} 
        = e^{-\int_0^t \sum_{j\neq i_0}A_j(x)\dd x}
        = e^{-h_{ot}\cdot t - 2(1 - x_0 - h_{ot}) \cdot (t - \threshold_1)^+}\eqdef e^{-K(t)} ~.
    \]
    
    Then, according to~\Cref{lem:acceptance-probability}, the probability of accepting item $i_0$ value $v_{i_0} = v$ is:
    \begin{align*}
        \int_0^1 p_{i_0}^v g_{i_0}^v(t) e^{-K(t)} \dd t = ~ & p_{i_0}^v \left( z_{i_0}^v \int_{\threshold_0}^1 e^{-K(t)} \dd t + (\rho_{i_0}^v - z_{i_0}^v) \cdot 2 \int_{\threshold_2}^1 e^{-K(t)} \dd t \right) \\
        \ge ~ & p_{i_0}^v \rho_{i_0}^v \cdot \min \left(\int_{\threshold_0}^1 e^{-K(t)}\dd t,
        ~ 2 \int_{\threshold_2}^1 e^{-K(t)} \dd t \right)~.
    \end{align*}

    \paragraph{Other Items.}
    We define $L(t)$ as the probability that item $i_0$ is activated before time $t$. Then, according to the definition of $g_{i_0}^v(t)$ and $h_0$, $\sum_v p_{i_0}^v g_{i_0}^v(t)$ equals $0$ for $t\in [0, \threshold_0]$, equals $h_0$ for $t\in [\threshold_0, \threshold_2]$, and equals $2 x_0 - h_0$ for $t\in [\threshold_2, 1]$. Thus, by~\Cref{lem:activation-probability}, we have
    \[
        \begin{aligned}
            L(t) = \int_0^t \sum_v p_{i_0}^v g_{i_0}^v(x) \dd x= h_0 \cdot (t - \threshold_0)^+ + 2 (x_0 - h_0) \cdot (t - \threshold_2)^+.
        \end{aligned}
    \]

Therefore, for any item $i\neq i_0$, according to~\Cref{lem:acceptance-probability}, the probability of accepting item $i$ with value $v$ is equal to:
\begin{align*}
    & \int_0^1 p_i^v a_i^v(t) e^{-\int_0^t \sum_i A_i(x) \dd x} (1-L(t)) \prod_{j\neq i, i_0} e^{-\int_0^t A_{j}(y)\dd y}\dd t \\
    = ~ & p_i^v\int_0^1 a_i^v(t)  \prod_{j\neq i_0} e^{-\int_0^t A_{j}(y)\dd y} (1-L(t)) \dd t \\
    = ~ & p_i^v\int_0^1 a_i^v(t) e^{-K(t)} (1-L(t)) \dd t \\
    = ~ & p_i^v \left( z_i^v\cdot   \int_0^1 e^{-K(t)}  (1- L(t)) \dd t +  (\rho_i^v-z_i^v) \cdot 2 \int_{\threshold_1}^1 e^{-K(t) }(1-L(t)) \dd t \right) \\
    \ge ~ & p_i^v \rho_i^v \cdot \min \left( \int_0^1 e^{-K(t)} (1-L(t)) \dd t,
    ~ 2 \int_{\threshold_1}^1 e^{-K(t)} \left(1 - L(t)\right) \dd t  \right)
    ~.
\end{align*}

\paragraph{Choice of Parameters.}
Therefore, we define:
\begin{equation}
\label{eqn:prophet-secretary-ratio}
\begin{aligned}
    \Gamma(x_0, h_0, \threshold_0, \threshold_1, \threshold_2) = \min \Bigg(& \int_{\threshold_0}^1 e^{-K(t)} \dd t,
    ~ 2\int_{\threshold_2}^1 e^{-K(t)} \dd t,\\
    & \int_0^1 e^{-K(t)} \left(1 - L(t)\right) \dd t,
~ 2 \int_{\threshold_1}^1 e^{-K(t)} \left(1 - L(t)\right) \dd t \Bigg)
~,
\end{aligned}
\end{equation}
which is a function with respect to $x_0, h_0, \threshold_0, \threshold_1, \threshold_2$. 
It remains to show that:

\begin{lemma}
\label{lem:verification_687}
For every $x_0 \in [0,1]$ and $h_0 \in [0,x_0]$, there exist $0\le \threshold_0\le \threshold_1\le \threshold_2\le 1$, such that $\Gamma( x_0, h_0, \threshold_0, \threshold_1, \threshold_2) > 0.686.$
\end{lemma}

We use a computer-aided search to numerically verify this lemma (see \Cref{app:program-verify}).
Then, taking such $\threshold_0, \threshold_1, \threshold_2$ in the algorithm, we get the $0.686$ competitive ratio by~\Cref{lem:prophet-secretary-sufficient-condition}.
\end{proof}


\section{Prophet Secretary Matching}
\subsection{Preliminaries}
\paragraph{Model.}
Let there be an edge-weighted bipartite type graph $G=(U,V,E,w)$, where $U$ corresponds to the set of offline vertices that are known in advance to the algorithm and $V$ corresponds to the set of all possible types of each online vertices.

There are $n$ online vertices. Each online vertex $v_i$ has a type $v \in V$ drawn independently from a known distribution $F_i$, and an arrival time of $t_i$ drawn uniformly from $[0,1]$. We use $\piv$ to denote the probability of $v_i = v$. When $v_i$ arrives at time $t_i$, its type is realized and the algorithm decides immediately whether to match $v_i$ to an unmatched offline vertex $u$ (if exists). The goal is to maximize the total weight of the selected matching and to compete against the expected maximum matching.

\paragraph{LP Relaxation.}
Consider the following linear program for stochastic matching by Gamlath, Kale, and Svensson~\cite{soda/GamlathKS19}.
\begin{align*}
\textbf{LP:} \qquad \max_{\{\xiuv\}}: \quad & \sum_{i,u,v} w_{uv} \cdot \xiuv \\
\text{subject to} : \quad & \sum_{u} \xiuv \le \piv & \forall i \in [n], v \in V \\
& \sum_{i} \sum_{v \in S} \xiuv \le 1 - \prod_{i} \left(1- \sum_{v\in S} \piv\right) & \forall u \in U, S \subseteq V
\end{align*}

\begin{fact}[\cite{soda/GamlathKS19}]
This linear program can be solved efficiently and its value is an upper bound of the expected optimum matching.
\end{fact}

Our algorithm starts by solving the linear program and throughout the section, we use $\left\{\xiuv\right\}$ to denote the optimal solution of the LP.
For each online vertex $u$, we consider the following quantities that shall be crucial for our algorithm design:
\[
\forall i \in [n], u \in U, \quad \xiu \eqdef \sum_v \xiuv, \quad \xu \eqdef \max_i \xiu, \quad \text{and} \quad h_u \eqdef \sum_{i,v} \left(2 \xiuv - \piv \right)^+~.
\]
According to the constraints of the LP and by \Cref{lem:hu_xu}, we have that $h_u \le h_2(\xu)$.
\paragraph{Assumption.} Without loss generality, we assume that the graph is $1$-regular. I.e., for every offline vertex $u \in U$, $\sum_{i,v} \xiuv = 1$ and for every $i \in [n]$, $\sum_{u,v} \xiuv = 1$.
This can be achieved by first adding dummy offline vertices and dummy online types to the type graph, and then introducing dummy online vertices to the instance. A similar pre-processing step is applied by Fu et al.~\cite{icalp/FuTWWZ21} (refer to Lemma 2 of \cite{icalp/FuTWWZ21} for a detailed construction). Finally, we simulate the random arrival of the dummy online vertices on our own.


\subsection{Overview of Our Result}
Our main result is a $0.641$-competitive algorithm for the prophet secretary matching problem, beating the $1-1/e\approx 0.632$ barrier for the first time.

We introduce two algorithms that we call 1) Multistage Activation-based Matching (MAM) and 2) Constant Activation Rate except One Large Vertex (CAR). 
\begin{itemize}
\item Our Multistage Activation-based Matching algorithm is a generalization of the algorithm by Yan~\cite{soda/Yan24} that works in the i.i.d. arrival setting. Our algorithm achieves the same competitive ratio of Yan for an arbitrary ``infinitesimal instance'', i.e., when $\xu = o(1)$ for every $u \in U$. In other words, an instance is infinitesimal if every online vertex $v_i$ contributes $o(1)$ to every offline vertex $u$ in the offline maximum matching. 
Crucially, our algorithm avoids the pre-processing step of Yan, which to our knowledge, is hard to generalize beyond the i.i.d. case.
Formally, our algorithm matches each edge $(u,v_i=v)$ with probability $\ratmsm(\xu)  \cdot \xiuv$, where $\ratmsm(\cdot)$ is a decreasing function with $\ratmsm(0) \approx 0.645$ and $\ratmsm(1) = 1-1/e$. This algorithm and its analysis (including the definition of $\ratmsm$) are provided in \Cref{subsec:msm}.

\item The second algorithm is built on the technique of random order contention resolution schemes. It is noticed by Fu et al.~\cite{icalp/FuTWWZ21} that the RCRS by Lee and Singla~\cite{esa/LeeS18} can be adapted to a $1-1/e$ competitive algorithm for the prophet secretary matching problem and the ratio is tight. Moreover, we observe that the worst-case instance for the RCRS based approach is the infinitesimal case when $\xu$'s are small. 
To this end, we refine the RCRS algorithm by Lee and Singla so that our algorithm matches each edge $(u,v_i=v)$ with probability $\ratrcrs(\xu)$, where $\ratrcrs(\cdot)$ is an increasing function with $\ratrcrs(0) = 1-1/e$ and $\ratrcrs(1) = \sqrt{3}-1\approx 0.732$. This algorithm and its analysis (including the definition of $\ratrcrs$) are provided in \Cref{subsec:rcrs}.
\end{itemize}

Though each of the two algorithms only achieves a competitive ratio of $1-1/e$ in the worst case, they nicely complement each other. Finally, we apply a randomization between the two algorithms.

\begin{theorem}
\label[theorem]{thm:secretary_matching}
A hybrid algorithm by running MAM with probability $0.8$ and running CPR with probability $0.2$ is at least $0.641$-competitive for prophet secretary matching.
\end{theorem}
\begin{proof}
By \Cref{thm:matching_msm} and \Cref{thm:matching_rcrs}, the hybrid algorithm matches each edge $(u,v_i=v)$ with probability
\[
\left( 0.8 \cdot \ratmsm(\xu) + 0.2 \cdot \ratrcrs(\xu) \right) \cdot \xiuv \ge 0.641 \cdot \xiuv, \quad \forall \xu \in [0,1]~,
\]
where the inequality is verified with computer assistance. Refer to \Cref{fig:ratios_matching}. Recall that $\sum_{i,u,v} w_{uv} \cdot \xiuv$ is an upper bound of the expected maximum matching. The property achieved directly implies a competitive ratio of $0.641$.
\end{proof}

\begin{figure}
  \centering
  \includegraphics[width=0.5\textwidth]{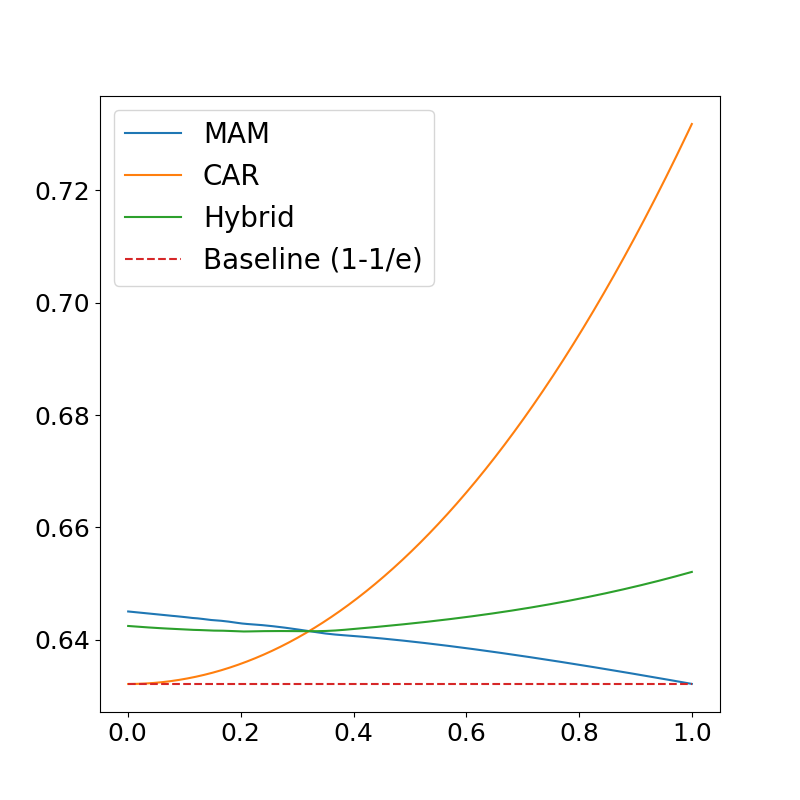}
  \caption{Competitive Ratio Curves For MAM, CAR, Hybrid}
  \label{fig:ratios_matching}
\end{figure}

\subsection{Query-Commit Stochastic Bipartite Matching}
Our algorithm for the prophet secretary matching can be applied to the query-commit stochastic bipartite matching problem. Consider a bipartite graph $G=(U,V,E)$ where each edge $e$ is associated with a weight $w_e$ and a probability $p_e$. 
At each step, the algorithm queries the existence of an edge $e$, which happens independently with probability $p_e$. If the edge exists, the algorithm must select it. The goal is to maximize the selected matching and competes against the optimum matching.

Gamlath, Kale, and Svensson~\cite{soda/GamlathKS19} designed an $1-1/e$ approximation algorithm for this problem, and the ratio is recently improved to $1-1/e+0.0014$ by Derakhshan and Farhadi~\cite{soda/Derakhshan023}.
As an implication of \Cref{thm:secretary_matching} through a folklore reduction~\cite{sicomp/FeldmanSZ21,icalp/FuTWWZ21} from the query-commit model to secretary model, we also improve the state-of-the-art approximation ratio of this problem. For completeness, we sketch a proof of the reduction in \Cref{app:reduction}.
\begin{theorem}
There exists a $0.641$ approximation algorithm for the query-commit stochastic bipartite matching problem.
\end{theorem}

\subsection{Multistage Activation-based Matching}
\label{subsec:msm}
In this section, we introduce the multistage activation-based matching algorithm. 
%
%
We first recall the correlated sampling method of Jaillet and Lu~\cite{mor/JailletL14}.
\paragraph{Jaillet and Lu Sampling\cite{mor/JailletL14}.} For every online vertex $v_i$ and type $v$:
\begin{itemize}
\item Consider an interval $[0,1]$ and align subintervals $I_u$ of length $\rhoiuv \eqdef \frac{\xiuv}{\piv}$ from left to right for each offline vertex $u \in U$. 
Refer to \Cref{fig:correlated-sample}. 
\item Sample $\eta$ uniformly at random from $[0,1]$. Let $\eta' = \eta \pm 1/2$ such that $\eta' \in [0,1]$.
\item Let $\muiv{u'}\eqdef \Prx{\eta \in I_u \text{ and } \eta' \in I_{u'}}$.
\end{itemize}
The quantities $\muiv{u'}$ are important parameters of our algorithm and analysis. We list the following properties of $\muiv{u'}$ that are straightforward to verify. 
\begin{lemma}
\label{lem:JL_sample}
	We have that
\begin{itemize}
\item $\muiv{u'} = \muiv[u']{u}$;
\item $\sum_{u' \ne u}\muiv{u'} = \rho_{i}^{(u,v)} - (2\rhoiuv-1)^+ \le \rhoiuv$.
\end{itemize}
\end{lemma}

\begin{figure*}
	\centering
	\begin{tikzpicture}
		\draw(0,0)--(2,0)[red,ultra thick];
		\draw(2,0)--(3.2,0)[blue,ultra thick];
		\draw(3.2,0)--(5,0)[green,ultra thick];
		\draw(0,-0.1)--(0,0.1);
		\draw(2,-0.1)--(2,0.1);
		\draw(3.2,-0.1)--(3.2,0.1);
		\draw(5,-0.1)--(5,0.1);
		\draw(1,0.3)node{$I_{u_1}$};
		\draw(2.6,0.3)node{$I_{u_2}$};
		\draw(4.1,0.3)node{$I_{u_3}$};
		\draw(1,-0.3)node{$0.4$};
		\draw(2.6,-0.3)node{$0.24$};
		\draw(4.1,-0.3)node{$0.36$};
		\draw(-5,2)--(-7,0)[red,thick];
		\draw(-5,2)--(-5,0)[blue,thick];
		\draw(-5,2)--(-3,0)[green,thick];
		\filldraw[fill=white](-5,2)circle(0.55)node{$v_i=v$};
		\filldraw[fill=white](-7,0)circle(0.3)node{$u_1$};
		\filldraw[fill=white](-5,0)circle(0.3)node{$u_2$};
		\filldraw[fill=white](-3,0)circle(0.3)node{$u_3$};
		\draw(-7,1.3)node[red]{$\rhoiuv[u_1]=0.4$};
		\draw(-5,0.6)node[blue]{$\rhoiuv[u_2]=0.24$};
		\draw(-3,1.3)node[green]{$\rhoiuv[u_3]=0.36$};
	\end{tikzpicture}
	\caption{Illustration of intervals $I_{u_1},I_{u_2}$ and $I_{u_3}$for $v_i=v$.}
	\label{fig:correlated-sample}
\end{figure*}
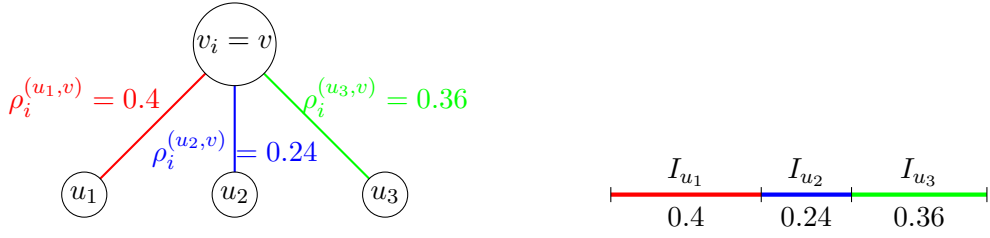

Now, we present our algorithm, a generalization of the algorithm by Yan~\cite{soda/Yan24}.

\begin{tcolorbox}[title=Multistage Activation-based Matching]
Parameters: $\beta_0=0.05, \beta_1=0.75$. \\
Notations: $M(t), N(t)$ are the sets of matched and unmatched offline vertices at time $t$.
\begin{enumerate}[(1)]
\item First Stage: for $t \in [0,\beta_0)$, let $\aiuv(t) = \ziuv$ for every $u,v,i$;
\item Second Stage: for $t \in [t_0,\beta_1)$, let $\aiuv(t) = \rhoiuv$ for every $u,v,i$;
\item If vertex $v_i$ arrives at time $t \in [0, \beta_1)$ with type $v \sim F_i$: \\
propose to a random offline vertex $u \in U$ with probability 
\begin{equation}
\giuv(t) = \aiuv(t) \cdot \exp\left(-\int_0^t \Aiu(y) \dd y\right),
\end{equation}
where $\Aiu(t) \eqdef \Exlong[v]{\aiuv(t)} = \sum_{v} \piv \cdot \aiuv(t)$. \\
Accept the edge $(u,v_i=v)$ if $u$ is unmatched.
\item Third Stage: at time $\beta_1$, let $\aiuv(t)$ for $t \in [\beta_1,1)$ be the following: \\
if $u$ is unmatched at time $t_1$, for every $v,i$,
\begin{multline*}
\aiuv(t) = \rhoiuv + \sum_{u' \ne u} \muiv[u']{u} \cdot \left(1 - \exp\left(-\int_0^{\beta_1} A_i^{u'}(y) \dd y\right) \right) \\
+ \sum_{u' \in M(\beta_1)} \muiv[u']{u} \cdot \exp\left(-\int_0^{\beta_1} A_i^{u'}(y) \dd y\right),
\end{multline*}
else, let $\aiuv(t) = 0$ for every $v,i$.
\item If vertex $v_i$ arrives at time $t \in [\beta_1,1]$ with type $v \sim F_i$: \\
propose to a random offline vertex $u \in U$ with probability 
\begin{equation}
\giuv(t) = \aiuv(t) \cdot \exp\left(-\int_0^t \Aiu(y)~\dd y\right),	
\end{equation}
where $\Aiu(t) \eqdef \Exlong[v]{2\rhoiuv - \ziuv} = \sum_{v} \piv \cdot \left(2\rhoiuv - \ziuv \right)$. \\
Accept the edge $(u,v_i=v)$ if $u$ is unmatched.
\end{enumerate}
\end{tcolorbox}

\clearpage
Our algorithm consists of three stages, in which the first two stages are non-adaptive and are similar to the step activation rates algorithm for the prophet secretary problem; the third stage is adaptive to the randomness of the instance and our algorithm from the first two stages.

The parameters $\aiuv(t)$ play a similar role as the activation rates for the prophet secretary problem. 
We first verify that the stated algorithm is well-defined. That is, the vector $\left\{\giuv(t)\right\}_{u}$ is a valid probability distribution for every $t \in [0,1]$.
\begin{lemma}
For every $v \in V, i \in [n], t \in [0,1]$, we have $\sum_{u} \giuv(t) \le 1$. 
\end{lemma}
\begin{proof}
For $t \in [0,\beta_1)$, we have
\[
\sum_{u} \giuv(t) \le \sum_{u} \aiuv(t)\le \sum_{u} \rhoiuv = 1~.
\]
For $t \in [\beta_1,1]$, notice that $\aiuv(t) = \aiuv(t_1)$ and $\giuv(t)$ is decreasing in $t$. It suffices to verify that $\sum_{u} \giuv(\beta_1) \le 1$.
We redistribute the terms with factor $\muiv[u']{u}$ in $a_i^{(u,v)}(t)$ to $u'$ and consider the (redistributed) contribution of every $u$ to the sum. If $u \in M(\beta_1)$, its contribution is:
\begin{multline*}
    \rhoiuv \cdot \exp\left(-\int_0^{\beta_1} \Aiu(y) \dd y\right) \\
    + \sum_{u'\neq u} \muiv[u]{u'} \cdot \left(1 - \exp\left(-\int_0^{\beta_1} \Aiu(y) \dd y\right) \right) \cdot \exp\left(-\int_0^{\beta_1} A_i^{u'}(y) \dd y\right) \\
    \le \rhoiuv \cdot \exp\left(-\int_0^{\beta_1} \Aiu(y) \dd y\right) + \sum_{u'\neq u} \muiv[u]{u'} \cdot \left(1 - \exp\left(-\int_0^{\beta_1} \Aiu(y) \dd y\right) \right) \\
    \le \rhoiuv \cdot \exp\left(-\int_0^{\beta_1} \Aiu(y) \dd y\right) + \rhoiuv \cdot \left(1 - \exp\left(-\int_0^{\beta_1} \Aiu(y) \dd y\right) \right) = \rhoiuv~,
\end{multline*}
where the first inequality is by dropping exponent term; the second inequality is by \Cref{lem:JL_sample}.
Else $u \in N(\beta_1)$, its contribution is:
\begin{align*}
    \sum_{u'\neq u} \muiv[u]{u'} \cdot \exp\left(-\int_0^{\beta_1} A_i^{u'}(y) \dd y\right)
    \le \sum_{u'\neq u} \muiv[u]{u'} \le \rhoiuv~.
\end{align*}
Therefore, we have $\sum_u \giuv(\beta_1) \le \sum_u \rhoiuv = 1$.
\end{proof}

For every vertex $u \in U$ and $t \in [0,1]$, let $A^u(t) \eqdef \sum_i \Aiu(t)$ and $\Amiu(t) \eqdef \sum_{j \ne i} A_j^u(t)$.
\begin{lemma}
\label{lem:aiu}
We have:
\[
A^u(t) =
\begin{cases}
\hu & t \in [0, \beta_0]~; \\
1 & t \in (\beta_0, \beta_1]~; \\
2-\hu & t \in (\beta_1, 1]~. \\
\end{cases}
\]
\end{lemma}
\begin{proof}
For $t\in [0,\beta_0]$, we have $A^u(t) = \sum_i \Aiu(t) = \sum_i \Ex[v]{\ziuv} = h_u$, where the last equality holds by the definition of $\hu$.
For $t \in (\beta_0,\beta_1]$, we have $A^u(t) = \sum_i \Ex[v]{\rhoiuv} = \sum_{i} \sum_v \piv \cdot \rhoiuv = \sum_{i,v} \xiuv = 1$.
For $t \in (\beta_1,1]$, we have $ A^u(t) = \sum_i \Ex[v]{2\rhoiuv - \ziuv} = 2-h_u$.
\end{proof}

As a consequence, we define and calculate the following functions:
\begin{itemize}
\item $c_1(h_u) \eqdef \frac{1-e^{-h_u\beta_0}}{h_u} = \int_0^{\beta_0} \exp\left( -\int_0^t A^u(y) \dd y \right) \dd t $~;
\item $c_2(h_u) \eqdef e^{-h_u \beta_0} \left( 1-e^{-(\beta_1-\beta_0)} \right) = \int_{\beta_0}^{\beta_1} \exp\left( -\int_0^t A^u(y) \dd y \right) \dd t $~;
\item $c_3(h_u) \eqdef e^{-h_u \beta_0 - (\beta_1 - \beta_0)} \cdot \frac{1-e^{-(2 - h_u) (1 - \beta_1)}}{2 - h_u} = \int_{\beta_1}^1 \exp\left( -\int_0^t A^u(y) \dd y \right) \dd t $~.
\end{itemize}

\begin{theorem}
\label[theorem]{thm:matching_msm}
For every $u \in U, v \in V, i \in [n]$, Multistage Proposal-based Matching matches edge $(u,v_i=v)$ with probability $\ratmsm(\xu) \cdot \xiuv$, where
\[
\ratmsm(x) \eqdef \min \left( c_1(h_2(x)) + c_2(h_2(x)) + c_3(h_2(x)), c_2(h_2(x)) + \left( 2-e^{-(\beta_1-\beta_0)} \right) \cdot c_3(h_2(x))\right)~.
\]
\end{theorem}
As an immediate corollary of the theorem, in the infinitesimal case when $\max_{u} x_u = o(1)$, our algorithm achieves a competitive ratio of $\ratmsm(0^+) = c_2(1-\ln 2) + (2 - e^{-(\beta_1 -\beta_0)}) c_3(1-\ln 2) \approx 0.645$. Note that we intentionally use the same parameters $\beta_0,\beta_1$ as Yan~\cite{soda/Yan24} so that our competitive ratio is exactly the same as Yan's ratio. Moreover, the competitive ratio of the algorithm is at least $\ratmsm(1) = 1-1/e$ for all instances.

\subsubsection{Proof of \Cref{thm:matching_msm}}

From now on, we fix arbitrary $u \in U, v \in V, i \in [n]$ and study the probability that $(u,v_i=v)$ is matched by our algorithm. 
We shall prove that
\begin{equation}
\label{eqn:matching_main}
\Prx{(u,v_i=v) \text{ is matched}} \ge \min \left( c_1(\hu) + c_2(\hu) + c_3(\hu), c_2(\hu) + \left( 2-e^{-(\beta_1-\beta_0)} \right) \cdot c_3(\hu)\right) \cdot \xiuv ~.
\end{equation}
This inequality together with the following lemma and that $h_u \le h_2(\xu)$ would conclude the statement of the theorem.
\begin{lemma}
\label{lem:math_fact}
The function $\left( c_1(h) + c_2(h) + c_3(h), c_2(h) + \left( 2-e^{-(\beta_1-\beta_0)} \right) \cdot c_3(h)\right)$ is decreasing for $h \in [0,1]$~.
\end{lemma}
We omit the tedious formal proof of this lemma and provide a plot of the function. Refer to \Cref{fig:monotonicity}.

\begin{figure}
  \centering
  \includegraphics[width=0.5\textwidth]{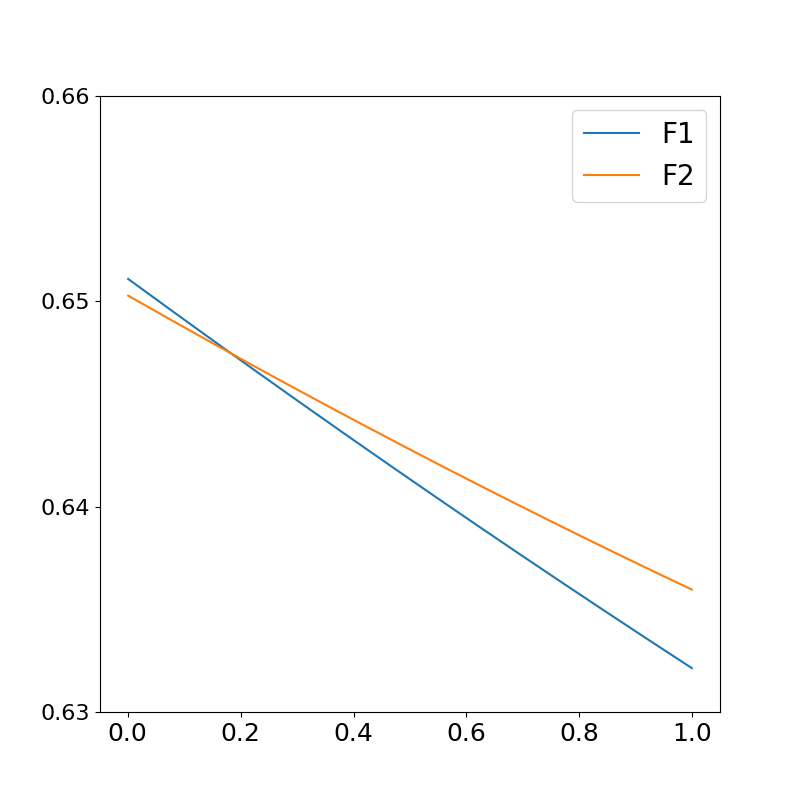}
  \caption{A Plot of $F_1(h) = c_1(h)+c_2(h)+c_3(h)$ and $F_2(h) = c_2(h)+(2-e^{-(\beta_1-\beta_0)}) c_3(h)$.}
  \label{fig:monotonicity}
\end{figure}

The next lemma is crucial for proving \Cref{eqn:matching_main}, that lower bounds the probability of $u$ remains unmatched on the arrival time $t_i=t$ of vertex $v_i$.
\begin{lemma}
\label{lem:matching_lb}
For every $t \in [0,\beta_1)$ and $u \in U$, 
\[
\Prx{u \in N(t) \mid t_i = t} = \exp\left( - \int_0^t \Amiu(y) \dd y \right)~.
\]
For every $t \in [\beta_1,1]$ and $u \ne u' \in U$,
\begin{align*}
& \Prx{u\in N(t) \mid t_i=t} \ge \exp\left( - \int_0^t \Amiu(y) \dd y \right)~; \\
& \Prx{u\in N(t), u'\in M(\beta_1) \mid t_i =t} \ge \exp\left( - \int_0^t \Amiu(y) \dd y \right) \left(1 - \exp\left( - \int_0^{t_1} \Amiu[u'](y) \dd y \right)\right) 
\end{align*}
\end{lemma}
\begin{proof}
For every $t\in [0, \beta_1)$ and $u\in U$, the statement and its proof is essentially the same of \Cref{lemma:single-act-rate} for the prophet secretary problem. Nevertheless, we include a proof for completeness.

\begin{multline*}
\Prx{u \in N(t) \mid t_i = t} = \prod_{j\neq i} \Prx{\text{$j$ doesn't propose $u$ before $t$}} = \prod_{j\neq i} \left(1 - \sum_v \piv\int_0^t \giuv(y) \dd y \right) \\
= \prod_{j\neq i} \left(1 - \int_0^t A_j^u(y)\cdot \exp\left(- \int_0^y A_j^u(z) \dd z \right) \dd y\right) \\
= \prod_{j\neq i} \exp\left(- \int_0^t A_i^u(y) \dd y\right) = \exp\left( - \int_0^t \Amiu(y) \dd y \right)~.
\end{multline*}
    
    For any $t\in [\beta_1, 1]$ and $u\neq u'\in U$, conditioned on box $i$ arrive exactly at $t$, we consider following hypothetical algorithm for other boxes and for time range $[1, t]$. with activation probability $\hat{g}_j^{(\cdot, v)}(t)$. It use the same activation method before $t_1$, which means $\hat{g}_j^{(\cdot, v)}(t) = g_j^{(\cdot, v)}(t)$, while for time $t$ after $t_1$, the algorithm proposed to $u$ with probability $\hat{g}_{j}^{(u,v)}(t) =  (2\rho_j^{(u,v)} - (2\rho_j^{(u,v)} - 1)^+) \cdot \exp\left( - \int_0^t A_j^u(y) \dd y \right)$, and do not propose to other vertices.

    Since the activation probability for vertex $u$ of original algorithm is at most the probability in this algorithm and the activation probability for vertex $u'$ is unchanged, the original probability of the event we concern is no smaller than the the probability of the same event while in this hypothetical algorithm. Following we work on the hypothetical algorithm, and lower bound the probability of the event.

    An observation is that for any $j\neq i$:
    \[\sum_v p_j^v \hat{g}_j^{(u,v)}(t) = A_j^u(t) \cdot \exp\left( - \int_0^t \Aiu(y) \dd y \right)~.\]

    Hence, the probability of $u$ is unmatched before $t$ is:
    \begin{align*}
        \prod_{j\neq i}\Pr[\text{$j$ doesn't propose $u$ before $t$}]
        = & \prod_{j\neq i} \left(1 - \sum_{j}p_i^v\int_0^{t} \hat{g}_i^{(u,v)}(y) \dd y \right)\\
        = & \prod_{j\neq i}\exp\left( - \int_0^t A_j^u(y) \dd y \right)~.
    \end{align*}

    The probability of $u$ is unmatched before $t$ and $u'$ is unmatched before $t_1$ is:
    \begin{align*}
        & \prod_{j\neq i}\Pr[\text{$j$ doesn't propose $u'$ before $t_1$ and $u'$ before $t$}]\\
        = ~ & \prod_{j\neq i} \left(1 - \sum_v \piv \int_0^{t} \hat{g}_i^{(u,v)}(y) \dd y - \sum_v \piv \int_0^{t_1} g_i^{(u',v)}(y) \dd y\right)\\
        \le ~ & \prod_{j\neq i} \left(1 - \sum_v \piv \int_0^{t} \hat{g}_i^{(u,v)}(y) \dd y \right) \left(1 - \sum_v \piv \int_0^{t_1} g_i^{(u',v)}(y) \dd y \right)\\
        = ~ & \prod_{j\neq i}\exp\left( - \int_0^t \Aiu(y) \dd y \right) \cdot \prod_{j\neq i}\exp\left( - \int_0^{t_1}A_j^{u'}(y)\dd y \right)~.
    \end{align*}

    The probability of $u$ is unmatched before $t$ and $u'$ is matched before $t_1$ is at least the difference of above two terms, as we desired.
\end{proof}

We first study the probability of matching $(u,v_i=v)$ in the first two stages of our algorithm, i.e., before time $\beta_1$. 
\begin{lemma}
The probability of $(u,v_i=v)$ being matched before time $\beta_1$ equals:
\begin{equation}
(2 \xiuv - \piv)^+ \cdot c_1(h_u) + \xiuv \cdot c_2(h_u)
\end{equation}
\end{lemma}
\begin{proof}
The probability of edge $(u,v_i=v)$ being matched in time $[0,\beta_1)$ is that
\begin{multline*}
\piv \cdot \int_0^{\beta_1} \giuv(t) \cdot \Prx{u \in N(t) \mid t_i =t} \dd t = \piv \cdot \int_0^{\beta_1} \giuv(t) \cdot \exp\left(-\int_0^{t} \Amiu(y)\dd y\right) \dd t \\
= \piv \cdot \int_0^{\beta_1} \aiuv(t) \cdot \exp\left(- \int_0^{t} A^u(y) \dd y \right) \dd t = (2 \xiuv - \piv)^+ \cdot  c_1(h_u) + \xiuv \cdot c_2(h_u)~.
\end{multline*}
Here, the first equality is by \Cref{lem:matching_lb}; the second and third equalities come from the construction of $\aiuv(t)$ and $c_1(h_u),c_2(h_u)$.
\end{proof}

Next, we study the probability of matching $(u,v_i=v)$ in the third stage of our algorithm, i.e., after time $\beta_1$.
\begin{lemma}
The probability of $(u,v_i=v)$ being matched between time $[\beta_1,1]$ is at least:
\[
\left( \xiuv + (\xiuv - (2 \xiuv - \piv)^+) \cdot  \left(1 - e^{- (\beta_1 - \beta_0)} \right) \right) \cdot c_3(h_u)~.
\]
\end{lemma}
\begin{proof}
The probability of the edge being matched in the third stage $[\beta_1, 1]$ equals the following:
\begin{multline}
\label{eqn:after_beta1}
\Exlong{\int_{\beta_1}^1 \piv \cdot \giuv(t) \cdot \ind{u \in N(t)}\dd t} \\
= \piv \cdot \Exlong{\int_{\beta_1}^{1}\aiuv(\beta_1) \cdot \ind{u \in N(t)} \cdot \exp\left(-\int_0^t \Aiu(y)\dd y\right)  \dd t} \\
= \piv \cdot \int_{\beta_1}^{1} \Exlong{\aiuv(\beta_1) \cdot \ind{u \in N(t)}} \cdot \exp\left(-\int_0^t \Aiu(y)\dd y\right)  \dd t~,
\end{multline}
where the expectations are taken over both the randomness from the instance (i.e., the realization of the values and the arrival times except for vertex $v_i$) and the randomness from our algorithm; the second equation follows from the fact that $\Aiu(y)$ are prefixed constants. 

Further, we have that
\begin{align*}
& \quad \Exlong{\aiuv(\beta_1) \cdot \ind{u \in N(t)}} = \Ex{\rhoiuv \cdot \ind{u \in N(t)}} \\
& + \sum_{u' \ne u} \Ex{\muiv[u']{u} \cdot \left(1-\exp\left(-\int_0^{\beta_1} \Aiu[u'](y) \dd y \right) \right) \cdot \ind{u \in N(t)}} \\
& + \sum_{u' \ne u} \Ex{\muiv[u']{u} \cdot \exp\left(-\int_0^{\beta_1} \Aiu[u'](y) \dd y\right) \cdot \ind{u \in N(t), u' \in M(\beta_1)}} \\
& \ge \exp\left( - \int_0^t \Amiu(y) \dd y \right) \cdot \left( \rhoiuv + \sum_{u' \ne u} \muiv[u']{u} \cdot  \left(1 - \exp\left( - \int_0^{\beta_1} A^{u'}(y) \dd y \right)\right) \right) \\
& \ge \exp\left( - \int_0^t \Amiu(y) \dd y \right) \cdot \left( \rhoiuv + \sum_{u' \ne u} \muiv[u']{u} \cdot \left(1 - e^{- (\beta_1 - \beta_0)} \right) \right)\\
& = \exp\left( - \int_0^t \Amiu(y) \dd y \right) \cdot \left( \rhoiuv + \left( \rhoiuv - \ziuv \right) \cdot  \left(1 - e^{- (\beta_1 - \beta_0)} \right) \right)~,
\end{align*}
where the first inequality follows from \Cref{lem:matching_lb}; the second inequality follows from the fact that 
$\exp \left( -\int_0^{\beta_1} A^{u'}(y) \dd y \right) \le \exp \left( - \int_{\beta_0}^{\beta_1} 1 \dd y\right) = e^{-(\beta_1-\beta_0)}$.
Finally, we have that
\begin{multline*}
\eqref{eqn:after_beta1} \ge \piv \cdot \left( \rhoiuv + (\rhoiuv - \ziuv) \cdot  \left(1 - e^{- (\beta_1 - \beta_0)} \right) \right) \cdot \int_{\beta_1}^{1} \exp\left(-\int_0^t A^u(y) \right) \dd y \\
= \left( \xiuv + (\xiuv - (2 \xiuv - \piv)^+) \cdot  \left(1 - e^{- (\beta_1 - \beta_0)} \right) \right) \cdot c_3(h_u)~.
\end{multline*}
\end{proof}

Putting the two lemmas together, we have that
\begin{multline*}
\Prx{(u,v_i=v) \text{ is matched}} \ge 
(2 \xiuv - \piv)^+ \cdot c_1(h_u)+ \xiuv \cdot c_2(h_u) \\
+ \left( \xiuv + (\xiuv - (2\xiuv - \piv)^+) \cdot  \left(1 - e^{- (\beta_1 - \beta_0)} \right) \right) \cdot c_3(h_u) \\
\ge \min \left( c_1(\hu) + c_2(\hu) + c_3(\hu), c_2(\hu) + \left( 2-e^{-(\beta_1-\beta_0)} \right) \cdot c_3(\hu)\right) \cdot \xiuv~,
\end{multline*}
where the last inequality holds by the fact that $\xiuv \ge (2\xiuv -\piv)^+$.

\subsection{Constant Activation Rate Except One Large Vertex}
\label{subsec:rcrs}
The multistage activation-based matching algorithm achieves an improved competitive ratio over $1-1/e$ on instances with small $\xu$'s, we are left to design an algorithm working against large $\xu$'s, which is relatively the easier task.
We adapt the constant activation rate algorithm for the prophet secretary problem to each offline vertex $u$, and at the same time apply a special treatment to the most important online neighbor of $u$. Our algorithm is as the following.

\begin{tcolorbox}[title=Constant Activation Rate Except One Large Vertex]
Parameters: for every $x \in [0,1]$, let $\alpha(x)$ be the unique solution to the following equation:
\[
\int_{\alpha(x)}^1 e^{-t(1-x)} \dd t =  \int_0^1 e^{-t(1-x)} \cdot \left( 1 - (t-\alpha(x))^+\cdot x \right) \dd t~.
\]
\begin{enumerate}[(1)]
\item For every $u \in U$, let $i_u \eqdef \argmax_i \xiu$ be the most important online vertex of $u$.
\item For time $t$ from $0$ to $1$:
\begin{enumerate}
\item If vertex $v_i$ arrives with type $v \sim F_i$: \\
propose to a random offline vertex $u \in U$ with probability $\rhoiuv$.
\item If vertex $u$ receives an proposal from vertex $v_i = v$:
\begin{itemize}
\item if $i = i_u$, activate the edge iff $t \ge \alpha(x_u)$;
\item if $i \ne i_u$, activate the edge with probability $e^{-t \cdot \xiu}$~.
\end{itemize}
Accept the first active edge that $u$ receives.
\end{enumerate}
\end{enumerate}
\end{tcolorbox}

\begin{theorem}
\label[theorem]{thm:matching_rcrs}
For every $u \in U, v \in V, i \in [n]$, the Constant Activation Rate Except One Large Vertex algorithm matches edge $(u,v_i=v)$ with probability $\ratrcrs(\xu) \cdot \xiuv$, where
\[
\ratrcrs(x) \eqdef \int_{\alpha(x)}^1 e^{-t(1-x)} \dd t\]
\end{theorem}
By straightforward calculation, $\ratrcrs(0) = 1-1/e$ and $\ratrcrs(1)=\sqrt{3}-1$. As an implication, this algorithm achieves a competitive ratio better than $1-1/e$ if every $\xu$ is bounded away from $0$.

\subsubsection{Proof of \Cref{thm:matching_rcrs}}
Fix an arbitrary $u \in U$. According to the definition of our algorithm, the probability that the edge $(u,v_j)$ is not active before time $t$:
\begin{align}
\label{eqn:rcrs_j_nonactive}
1 - \int_0^t \sum_{v} \left( \pjv \cdot \rhojuv \cdot e^{-y \xju} \right)\dd y = 1 - \int_0^t \xju \cdot e^{-y \xju} \dd y = e^{-t \xju}~, \quad \forall j \ne i_u~; \\
\label{eqn:rcrs_iu_nonactive}
1 - \int_0^t \ind{y \ge \alpha(\xu)} \cdot \sum_{v} \pjv \cdot \rhojuv \dd y = 1-(t-\alpha(\xu))^+\cdot \xu~, \quad j=i_u~;
\end{align}
Next, we consider probability that edge $(u,v_i=v)$ being matched.

\paragraph{Case 1. ($i = i_u$)} 
The edge $(u,v_i=v)$ is matched with probability:
\begin{multline*}
\int_{\alpha(\xu)}^1 \piv\cdot \rhoiuv \cdot \Prx{u \in N(t) \mid t_i = t} \dd t = \xiuv \cdot \int_{\alpha(\xu)}^1  \prod_{j\ne i} e^{-t \xju} \dd t \ge \xiuv \cdot \int_{\alpha(\xu)}^{1} e^{-t(1-\xu)} \dd t~,
\end{multline*}
where the equality is by \eqref{eqn:rcrs_j_nonactive} and the inequality holds by the fact that $\sum_{j \ne i} \xju \le 1-\xiu=1-\xu$.

\paragraph{Case 2. ($i \ne i_u$)} 
The edge $(u,v_i=v)$ is matched with probability:
\begin{multline*}
\int_{0}^1 \piv\cdot \rhoiuv \cdot \Prx{u \in N(t) \mid t_i = t} \dd t = \xiuv \cdot \int_{0}^1  \prod_{j\ne i, i_u} e^{-t \xju} \cdot \left( 1-(t-\alpha(\xu))^+\cdot \xu \right)\dd t \\
\ge \xiuv \cdot \int_{0}^1  e^{-t(1-\xu)} \cdot \left( 1-(t-\alpha(\xu))^+\cdot \xu \right)\dd t = \xiuv \cdot \int_{\alpha(\xu)}^{1} e^{-t(1-\xu)} \dd t
\end{multline*}
where the first equality is by \eqref{eqn:rcrs_j_nonactive} and \eqref{eqn:rcrs_iu_nonactive}, and the last equality is by the definition of $\alpha(\cdot)$.

This concludes the proof of the theorem.

\bibliography{matching,ldc_bib}

\appendix

\section{Reduction from Query-Commit to Prophet Secretary Matching}
\label{app:reduction}

First of all, notice that the stochastic model of the underlying graph of query-commit is a special case of the model of prophet secretary matching. Indeed, the prophet secretary problem only assumes independent realizations of vertices, while query-commit further assumes independent realizations of edges.

Therefore, we can think of an arbitrary instance for query-commit to be an instance for prophet secretary by artificially separating the two sides of the vertices to be online and offline. Say $V$ is the set of online vertices and $U$ is the set of offline vertices. 
In order to apply an online algorithm $A_{PSM}$ of the prophet secretary matching problem for the instance, we draw a uniform at random arrival time $t_v$ for each vertex $v \in V$ on our own to simulate the random arrival of the vertices.
Then, the difficulty is that we are not able to see the realization of $v$'s type before we query the existence of an edge.

Nevertheless, on the arrival of an online vertex $v$, we first compute the expected matching probability $x_v^u$ of vertex $v$ to any offline vertex $u$, according to algorithm $A_{PSM}$ in an imaginary run of the prophet secretary problem. 
The vector $(x_v^u)_{u\in U}$ then must be in the following polytope:
\begin{align*}
    & \sum_{u\in S} x_v^u \le 1 - \prod_{u\in S}\left(1-p_{(v,u)}\right), && \forall S\subseteq U,\\
    & x_v^u \ge 0, && \forall u\in U~,
\end{align*}
where the first inequality states that the probability of $v$ being matched to some $u \in S$ is no larger than the probability that at least one of the edges $\{(v,u)\}_{u\in S}$ exists. 
We import the following lemma by \citet{soda/GamlathKS19} which characterizes the extreme points of this polytope.

\begin{lemma}[Lemma 5 of \cite{soda/GamlathKS19} rephrased]
    Every vertex of above polytope corresponds to a permutation of a subset of the offline vertices $u_1, u_2, \dots, u_\ell$ such that:
    \begin{align*}
    x_v^{u_1} & = p_{(v,u_1)} \\
    x_v^{u_2} & = \left(1-p_{(v,u_1)}\right) p_{(v,u_2)}  \\
    \dots & \\
    x_v^{u_\ell} & = (1-p_{(v,u_1)})\dots(1-p_{(v,u_{\ell-1})})p_{(v,u_\ell)} \\
    x_v^{u} & = 0 \quad \text{if $u \ne u_1, u_2, \dots, u_\ell$}
    \end{align*}
\end{lemma}

Notice that querying the vertices in the order of $u_1, u_2, \dots, u_\ell$, the vector of the expected probability exactly equals to the extreme point showed above. Thus, we can match the expected probability $(x_v^u)_{u\in U}$ in the query-commit model through the following: 1) decompose $(x_v^u)_{u\in U}$ as a convex combination of the extreme points and sample from the extreme points, 2) then query the offline vertices in the order with respect to the sampled extreme point.
Through this simulation, our algorithm matches the same performance of $A_{PSM}$. Finally notice that the benchmark (expected offline optimum matching) for the two problems are also the same.

\section{Proof of \Cref{lem:h_s_general}}
\label{app:h-function}

We study the following optimization problem and define $h_s(x)$ to be the optimal value:
    \begin{align*}
        \max_{\{x_i^v\},\{p_i^v\}}: \quad &\frac{1}{s-1} \cdot \sum_{i,v} \left(sx_i^v - p_i^v \right)^+\\
        \text{subject to}: \quad & \sum_{i} \sum_{v\in V_i}x_i^v \le 1 - \prod_i \left(1 - \sum_{v\in V_i}p_i^v \right) && \forall V_i \subseteq V\\
        & \sum_v x_i^v\le x && \forall i\in [n]
    \end{align*}

First, we reduce the dimension of the optimization by showing that it is without loss of generality to study when $|V|=1$. Specifically, we show that $h_s(x)$ equals to the value of the following program:
    \begin{align*}
        \max_{\{x_i\},\{p_i\}}: \quad & \frac{1}{s-1} \cdot \sum_{i} (sx_i - p_i)\\
        \text{subject to}: \quad & \sum_{i\in S} x_i \le 1 - \prod_{i\in S} (1 - p_i) && \forall S \subseteq [n]\\
        & x_i\le x && \forall i\in [n]
    \end{align*}
It is obvious that this program is a special case of the original one. 
On the other hand, given an arbitrary feasible solution $\{x_i^v\}, \{p_i^v\}$ of the original program. Consider the following vectors:
\[
x_i = \sum_{v:s x_i^v - p_i^v\ge 0} x_i^v \quad \text{and} \quad p_i = \sum_{v:s x_i^v - p_i^v\ge 0} p_i^v~,\quad \forall i\in [n]
\]
It is straightforward to verity that the constraints of the new program are satisfied, and the objective value equals to the objective of the original program.


Next, we fix the vector $\{x_i\}$ and consider the program as an optimization over $\alpha_i = \ln(1-p_i)$. The constraints can be rewritten as the following:
\begin{equation}
\label{eqn:poly_alpha}
\sum_{i\in S} \alpha_i \le \ln \left( 1 - \sum_{i\in S}x_i \right), \quad \forall S\subseteq [n]
\end{equation}
And the goal is now to minimize $\sum_i 1 - e^{\alpha_i}$, which is a concave function of $\{\alpha_i\}$.
Thus, it attains its minimum value at the extreme point of above polytope \eqref{eqn:poly_alpha}.

Let $\{\alpha_i\}$ be an extreme point. Then, there are $n$ tight constraints, corresponding to $n$ sets $S_1,\ldots,S_n$. 
We claim that the $n$ sets must form a chain 
\[
S_1\subset S_2\subset \cdots \subset S_n.
\]
Otherwise suppose there are two sets $S,T$ with $S\not\subset T$ and $T\not\subset S$ while the two corresponding constraints are tight. Then we have
\begin{align*}
    & \ln\left(1 - \sum_{i\in S} x_i\right) + \ln\left(1 - \sum_{i\in T} x_i\right) = \sum_{i\in S} \alpha_i + \sum_{i\in T} \alpha_i \\
    = ~ & \sum_{i\in S\cap T} \alpha_i + \sum_{i\in S\cup T}\alpha_i \le \ln\left(1 - \sum_{i\in S\cap T} x_i\right) + \ln\left(1 - \sum_{i\in S\cup T} x_i\right),
\end{align*}
a contradiction due to the concavity of $\ln(1-x)$.
Furthermore, since $|S_n| \le n$, we must have that $|S_i| = i$ for $i \in [n]$. 
Without loss of generality, we reorder the indices and assume that $S_i = [i]$. According to those tight constraints, we have 
\[
    \alpha_i = \ln\left(\frac{1 - \sum_{j\le i} x_j}{1 - \sum_{j<i}x_j}\right) \quad \text{and} \quad p_i = \frac{x_i}{1 - \sum_{j < i}x_j}~.
\]
We simplify the program as the following.
\begin{align*}
    \max_{\{x_i\}}: \quad & \frac{1}{s-1} \cdot  \sum_{i} \left(s x_i - \frac{x_i}{1 - \sum_{j < i}x_j} \right) \\
    \text{subject to}: \quad & \sum_{i} x_i \le 1 \\
        & x_i \le x && \forall i \in [n]
    \end{align*}
Let $\{x_i\}$ be the optimal solution. We prove that $x_i$'s are non-decreasing by contradiction. 
Suppose $x_i > x_{i+1}$. Consider switching the order of them, the objective function strictly improves since
\begin{align*}
& \frac{x_i}{1 - y} + \frac{x_{i+1}}{1 - y - x_i} > \frac{x_{i+1}}{1 - y} + \frac{x_i}{1 - y - x_{i+1}} \\
\iff & \frac{x_i - x_{i+1}}{1 - y} \ge \frac{(1 - y - x_i)x_i - (1 - y - x_{i+1})x_{i+1}}{(1 - y - x_i)(1 - y - x_{i+1})} \\
\iff & \frac{1}{1-y} > \frac{1 - y - x_i - x_{i+1}}{(1 - y - x_i) (1 - y - x_{i+1})}~,
\end{align*}
which holds for every $x_i > x_{i+1}$ and $y = \sum_{j<i} x_j$.
Thus, we have $x_1 \le x_2 ... \le x_n$. 
Finally, we prove that there is at most one $x_i$ in the open interval $(0, x)$ by contradiction. 
Suppose $0 < x_i \le x_{i+1} < x$. Consider changing $(x_i, x_{i+1})$ to $(x_i - \epsilon, x_{i+1} + \epsilon)$, the objective strictly improves since
\begin{align*}
& \frac{\partial}{\partial \epsilon} \left(\frac{x_i-\epsilon}{1 - y} + \frac{x_{i+1}+\epsilon}{1 - y - x_i + \epsilon} \right) < 0 \\
\iff & \frac{1-y-x_i-x_{i+1}}{(1-y-x_i+\epsilon)^2} < \frac{1}{1-y} \\
\Longleftarrow & \frac{1-y-x_i-x_{i+1}}{(1-y-x_i)^2} < \frac{1}{1-y}~,
\end{align*}
which holds for every $x_i \le x_{i+1}$ and $y = \sum_{j<i} x_j$.
To sum up, the optimal solution has the form of $0 \le x_1 < x_2 = x_3 = ... = x_n = x$ and the objective is
\[
x_1 + \frac{1}{s-1} \cdot \sum_{k=0}^{n-1} \left( s x - \frac{x}{1-x_1-i x} \right) \le \max_{t\in [0,x)} t + \frac{1}{s-1}\sum_{k \ge 0} \left(sx - \frac{x}{1 - t - kx}\right)^+~,
\]
that concludes the proof of the lemma.

\subsection{Some properties of $h_s(x)$}
As promised in \Cref{lem:hu_xu}, we verify that $h_2(0^+) = 1-\ln 2$.
\begin{lemma}
    For any $s>1$, $\lim_{x \to 0^+} h_s(x) = 1 - \frac{\ln s}{s-1}$.
\end{lemma}
\begin{proof}
    Notice that $h_s(x)$ is a Riemann sum for the integral of function $f_s(y) = \frac{1}{s-1} \left( s-\frac{1}{1 - y} \right)^+$ on $[0, 1 - \frac1s]$, with intervals smaller than $x$. Hence
    \[
        \lim_{x\to 0^+} h_s(x) = \frac{1}{s-1} \cdot \int_0^{1 - \frac1s} \left(s - \frac{1}{1-y}\right) \dd y = 1 - \frac{\ln s}{s-1}~.
    \]
\end{proof}

\begin{lemma}
    For any $s>1$ and any $x\in [0,1)$, $h_s(x)$ is right continuous.
\end{lemma}
Actually, $h_s(x)$ is also left continuous, but we will conclude a stronger property for the left side in the next lemma.
\begin{proof}
For any $x\in (0,1)$, and any $x'>x$,
suppose that 
\[
    h_s(x') = t + \frac{1}{s-1} \sum_{k\in \mathbb{N}, kx'+t\le 1 - 1/s} \left(sx' - \frac{x'}{1 - t - kx'} \right).
\]
Let $n\in \mathbb{N}$ be that $(n-1)x + t < 1 - 1/s \le nx + t$, then $n$ is bounded according to $x$, $n < (1-1/s)/x + 1$. The range of the summation can be rewritten as from $0$ to $n-1$. Then, for any $x'\in (t, x)$, we have
\[
    h_s(x) \ge \min\{t, x\} + \frac{1}{s-1} \sum_{k = 0}^{n-1} \left(sx - \frac{x}{1 - t - kx}\right)~.
\]
Hence,
\begin{align*}
    & h_s(x') - h_s(x) \le x' - x + \frac{1}{s-1} \sum_{k = 0}^{n-1} (sx' - sx) \\
    =~ & (1 + \frac{sn}{s-1})(x'-x) < (1 + \frac1x + \frac{s}{s-1})(x'-x)
\end{align*}
and our lemma follows.
\end{proof}

\begin{lemma}
    For any $s>1$ and any $x\in (0,1]$, the left derivative of $h_s(x)$ is upper bounded by $\frac{s+1}{2}$,
    \[
        \frac{\partial h_s(x)}{\partial x^-} \le \frac{s + 1}{2}.
    \]
\end{lemma}
\begin{proof}
Suppose that 
\[
    h_s(x) = t + \frac{1}{s-1} \sum_{k\in \mathbb{N}, kx+t\le 1 - 1/s} \left(sx - \frac{x}{1 - t - kx} \right).
\]
Let $n\in \mathbb{N}$ be that $(n-1)x + t < 1 - 1/s \le nx + t$. The range of the summation can be rewritten as from $0$ to $n-1$.
Then, for any $x'\in (t, x)$, we have
\[
    h_s(x') \ge t + \frac{1}{s-1} \sum_{k = 0}^{n-1} \left(sx' - \frac{x'}{1 - t - kx'}\right)~.
\]
Therefore therefore the left derivative of $h_s(x)$ is at most:
\[
    \frac{1}{s-1} \sum_{k = 0}^{n-1} \left(s - \frac{1 - t}{(1 - t - kx)^2}\right)~.
\]
This formula is non-increasing with respect to $x$. Without loss of generality, we assume $nx + t = 1 - 1/s$.
Consider the following equation:
\begin{align*}
    \sum_{k = 0}^{n-1} \frac{1-t}{(1-t-kx)(1-t-(k+1)x)} 
    & = \frac{1-t}{x} \sum_{k = 0}^{n-1} \frac{1}{1-t-(k+1)x} - \frac{1}{1-t-kx}\\
    & = \frac{1-t}{x} \cdot \left(s - \frac{1}{1-t}\right) =  sn.
\end{align*}
Finally, we have
\begin{align*}
    (s-1) \cdot \frac{\partial h_s(x)}{\partial x^-}
    & \le (1-t) \cdot \sum_{k = 0}^{n-1} \frac{1}{(1-t-kx)(1-t-(k+1)x)} - \frac{1}{(1-t-kx)^2}\\
    & \le (1-t) \cdot \sum_{k = 0}^{n-1} \frac12 \cdot \frac{1}{(1-t-(k+1)x)^2} - \frac12 \cdot \frac{1}{(1-t-kx)^2}\\
    & = (1-t) \cdot \left( \frac{s^2}{2} - \frac{1}{2(1-t)^2} \right) \le \frac{s^2 - 1}{2}.
\end{align*}
where the second inequality is by AM-GM inequality.
\end{proof}

Now, we are able to globally bound the difference of $h_s(x)$ and $h_s(x')$ by their difference $x-x'$ within a constant factor, that shall be used in \Cref{app:program-verify}.
\begin{lemma}
\label{lem:h-derivative}
    For any $s>1$ and any $x' \le x\in [0, 1]$ we have $h_s(x)-h_s(x')\le \frac{s+1}{2}(x-x')$.
\end{lemma}

\begin{proof}
    For any $x\in [0,1]$, consider set $A = \{a : \forall b\in [a,x], h_s(x) - h_s(b)\le \frac{s+1}{2}(x-b), 0\le a \le x\}$. This set is non-empty since $x\in A$. We denote $\inf A$ by $a$.

    First we prove $a\in \inf A$. Since $a = \inf A$, we know that $h_s(x) - h_s(b)\le \frac{s+1}{2}(x-b)$ for any $a<b\le x$. Take limit $b\to a^+$ on both sides gives us $h_s(x) - h_s(a)\le \frac{s+1}{2}(x-a)$, so $a\in A$.

    Next we prove $a = 0$. Otherwise suppose that $a>0$, since the left derivative of $a$ is upper bounded by $\frac{s+1}{2}$, there exist $0<t<a$ such that for any $b\in (t,a)$, $h_s(a) - h_s(b) \le \frac{s+1}{2}(a-b)$, which implies that $h_s(x) - h_s(b) \le \frac{s+1}{2}(x-b)$. As a result, for any $b\in (t,a)$, $b\in A$, contradict to $a = \inf A$. Therefore we must have $a=0$.

    Hence, $A = [0,x]$ and our lemma follows.
\end{proof}
\newpage


\section{Algorithm with general $s$}
\label{app:algo-general-s}

\begin{tcolorbox}[title=Step Activation Rates Except One Large Item]
    \label{algo:general-step-act-rate-0.688}
    \begin{enumerate}[(1)]
        \item Let $i_0=\arg \max_i \sum_v x_i^v$ be the largest item, and let:
        	\[
        		x_0 = \sum_v x_{i_0}^v
        		\quad,\quad
        		z_{i_0}^v = \frac{(s \rho_{i_0}^v - 1)^+}{s-1}
        		\quad,\quad
        		h_0 = \sum_{v} p_{i_0}^v z_{i_0}^v
        		~.
    		\]
        \item For calculate a set of $z_i^v$'s for the remaining items $i \ne i_0$ such that:
			\[
				h_{ot} \eqdef \sum_{i\neq i_0} \sum_v p_i^v z_i^v = \min\{h_s(x_0) - h_0, 1 - x_0\}
				\quad,\quad
				z_i^v \in \left[ \frac{1}{s-1}(s\rho_i^v - 1)^+, \rho_i^v \right]
				~,
			\]
			where $h_s(x_0)$ is the function as defined in \Cref{lem:h_s_general}.
        \item Select threshold times $\threshold_0 \leq \threshold_1 \leq \threshold_2$ based on the instance.
        \item For time $t$ from $0$ to $1$:
        \begin{enumerate}[(a)]
            \item If the largest item $i_0$ arrives with value $v$, activate it with probability:
            \[
	            g_{i_0}^v(t)=\left\{ \begin{aligned}
	                0 & \quad \text{if } t\in [0,\threshold_0),\\
	                z_{i_0}^v & \quad \text{if } t\in [\threshold_0, \threshold_2),\\
	                s\rho_{i_0}^v - (s-1)z_{i_0}^v & \quad \text{if } t\in [\threshold_2, 1].
	            \end{aligned} \right.
            \]
            \item If an item $i\neq i_0$ arrives with value $v$, activate it with activation rate:
            \[
            a_i^v(t)=\left\{ \begin{aligned}
                z_i^v & \quad \text{if } t\in [0,\threshold_1),\\
                s\rho_i^v - (s-1)z_i^v & \quad \text{if } t\in [\threshold_1, 1].
            \end{aligned} \right.
            \]
            \item Accept the item if it is the first activated item.
        \end{enumerate}
    \end{enumerate}
\end{tcolorbox}    



\begin{lemma}
	There exist $z_i^v$'s satisfying the stated constraints in (2).
\end{lemma}

\begin{proof}
    Note that $ \frac{1}{s-1}(s\rho_i^v - 1)^+ \leq \rho_i^v$, it suffices to show that
    \[
        \sum_{i\neq i_0} \sum_v p_i^v \cdot \frac{1}{s-1} (s\rho_i^v - 1)^+ \leq \min\{h_s(x_0) - h_0, 1 - x_0\} \leq \sum_{i\neq i_0} \sum_v p_i^v \rho_i^v.
    \]
    
    Recall that $h_s(x_0)$ is an upper bound of $\sum_{i,v} \frac{(sx_i^v - p_i^v)^+}{s-1}$ (see \Cref{app:h-function}).
    The first inequality follows by:
    \[
        \begin{aligned}
            \sum_{i\neq i_0} \sum_v p_i^v \cdot \frac{1}{s-1}(s \rho_{i}^v - 1)^+
            &
            = \sum_{i\neq i_0} \sum_v \frac{1}{s-1} (sx_i^v - p_i^v)^+ \\
            & 
            = \sum_{i} \sum_v \frac{1}{s-1} (sx_i^v - p_i^v)^+ -h_0  \leq h_s(x_0) - h_0
            ~.
        \end{aligned}
    \]
    
    The second inequality follows by $\sum_{i\neq i_0} \sum_v p_i^v \rho_i^v = \sum_{i\neq i_0} \sum_v x_i^v = 1 - x_0$.
\end{proof}

\begin{lemma}
    The activation rates and probabilities are well-defined, namely they are between $0$ and $1$.
\end{lemma}

\begin{proof} 
    Recall that $0\leq x_i^v \leq \rho_{i}^v\leq 1$.
    For any item $i$, we have $0 \le z_{i}^v \leq \rho_{i}^v \leq 1$, and $s \rho_{i_0}^v - (s-1) z_{i_0}^v= s \rho_{i_0}^v - (s \rho_{i_0}^v - 1)^+ = \min\{s \rho_{i_0}^v, 1\}\in [0,1]$
    
    Hence, for the largest item $i_0$, the activation probability $g_{i_0}^v(t)$ is always between $0$ and $1$.
    For any other item $i \ne i_0$, the activation rate $a_i^v(t)$ is between $0$ and $1$ and the resulting activation probability $g_i^v(t) = a_i^v(t)e^{-\int_0^t A_i(x) \dd x}$ is also between $0$ and $1$.
\end{proof}


Then, we are ready to prove the main theorem of this section~\Cref{thm:prophet-secretary}, which says that the Step Activation Rates Except One Large Item is $0.688$-competitive.

\begin{proof}[Proof of \Cref{thm:prophet-secretary}]
    Actually, we do not give the explicit form of $\threshold_0,\threshold_1,\threshold_2$. We only show that there always \emph{exists} proper choices of them based on any $x_0, h_0$ such that the competitive ratio is $0.688$. In the following, we let $\threshold_0\leq \threshold_1\leq \threshold_2$ to be determined.

    \paragraph{Largest Item.}
    By \Cref{lemma:single-act-rate}, the probability that the item $i(\neq i_0)$ was not activated before time $t$ is $e^{-\int_{0}^t A_i(x)\dd x}$, where $A_i(x)$ equals $\sum_v p_i^v z_i^v$ when $x\in [0, \threshold_1]$ and equals $s\sum_v x_i^v - \sum_v p_i^v z_i^v$ otherwise.

    By the definition of $h_{ot}$ and $\sum_j \sum_v x_j^v =1$ (\Cref{lem:prophet-secretary-properties}), we have:
    \[
        \sum_{j\neq i_0} A_j(x) = 
        \begin{cases} 
        \sum_{j \neq i_0} \sum_v p_j^v z_j^v = h_{ot}, & \text{for } x \in [0, \threshold_1] \\
        \sum_{j \neq i_0} (s\sum_v x_j^v - \sum_v p_j^v z_j^v) = s(1 - x_0) - h_{ot}. & \text{for } x \in [\threshold_1, 1]
        \end{cases}
    \]
    Thus, the probability of not existing any item $j\neq i_0$ was activated before time threshold $t$ is:
    \[
    \begin{aligned}
        \prod_{j\neq i_0} e^{-\int_0^t A_j(x)\dd x} 
        = e^{-\int_0^t \sum_{j\neq i_0}A_j(x)\dd x}
        = e^{-h_{ot}\cdot t - s(1 - x_0 - h_{ot}) \cdot (t - \threshold_1)^+}\eqdef e^{-K(t)},
    \end{aligned}    
    \]
    Then, according to~\Cref{lem:acceptance-probability}, the probability of item $i_0$ is accepted with value $v_{i_0} = v$ is:
    \begin{align*}
        & \int_0^1 p_{i_0}^v g_{i_0}^v(t) e^{-K(t)} \dd t \\
        = ~ & p_{i_0}^v \left( z_{i_0}^v \int_{\threshold_0}^1 e^{-K(t)} \dd t + (\rho_{i_0}^v - z_{i_0}^v) \cdot s \int_{\threshold_2}^1 e^{-K(t)} \dd t \right) \\
        \ge ~ & p_{i_0}^v \rho_{i_0}^v \cdot \min \left(\int_{\threshold_0}^1 e^{-K(t)}\dd t,
        ~ s \int_{\threshold_2}^1 e^{-K(t)} \dd t \right)
    \end{align*}

    \paragraph{Other Items.}
    We define $L(t)$ as the probability that item $i_0$ is activated before time $t$. Then, according to the definition of $g_{i_0}^v(t)$ and $h_0$, $\sum_v p_{i_0}^v g_{i_0}^v(t)$ equals $0$ for $t\in [0, \threshold_0]$, equals $h_0$ for $t\in [\threshold_0, \threshold_2]$, and equals $s x_0 - h_0$ for $t\in [\threshold_2, 1]$. Thus, by~\Cref{lem:activation-probability}, we have
    \[
        \begin{aligned}
            L(t) = \int_0^t \sum_v p_{i_0}^v g_{i_0}^v(x) \dd x= h_0 \cdot (t - \threshold_0)^+ + s (x_0 - h_0) \cdot (t - \threshold_2)^+.
        \end{aligned}
    \]

Thus, for any item $i\neq i_0$, according to~\Cref{lem:acceptance-probability}, the probability of item $i$ to be accepted with value $v$ is:
\begin{align*}
    & \int_0^1 p_i^v a_i^v(t) e^{-\int_0^t \sum_i A_i(x) \dd x}  \cdot (1-L(t)) \prod_{j\neq i, i_0} e^{-\int_0^t A_{j}(y)\dd y}\dd t \\
    = ~ & p_i^v\int_0^1 a_i^v(t)  \prod_{j\neq i_0} e^{-\int_0^t A_{j}(y)\dd y} (1-L(t)) \dd t \\
    = ~ & p_i^v\int_0^1 a_i^v(t) e^{-K(t)} (1-L(t)) \dd t \\
    = ~ & p_i^v \left( z_i^v\cdot   \int_0^1 e^{-K(t)}  (1- L(t)) \dd t +  (\rho_i^v-z_i^v) \cdot s \int_{\threshold_1}^1 e^{-K(t) }(1-L(t)) \dd t \right)\\
    \ge ~ & p_i^v \rho_i^v \min\left( \int_0^1 e^{-K(t)} (1-L(t)) \dd t,
    ~ s \int_{\threshold_1}^1 e^{-K(t)} \left(1 - L(t)\right) \dd t \right)
\end{align*}

\paragraph{Choice of Parameters.}
Therefore, we define:
\begin{align}
    \Gamma(x_0, h_0, s, \threshold_0, \threshold_1, \threshold_2) = \min \Bigg( & \int_{\threshold_0}^1 e^{-K(t)} \dd t,
    ~ s\int_{\threshold_2}^1 e^{-K(t)} \dd t,\\
    & \int_0^1 e^{-K(t)} \left(1 - L(t)\right) \dd t,
~ s \int_{\threshold_1}^1 e^{-K(t)} \left(1 - L(t)\right) \dd t \Bigg)
~,
\end{align}

We use computer to numerically verify the following, that we discuss in detail in \Cref{app:program-verify}.
\begin{lemma}
\label{lem:verification_689}
For every $x_0 \in [0,1]$ and $h_0 \in [0,x_0]$, there exists $s \ge 1, 0\le \threshold_0\le \threshold_1\le \threshold_2\le 1$, such that $\Gamma(x_0, h_0, s, \threshold_0, \threshold_1, \threshold_2) > 0.688.$
\end{lemma}
This concludes a competitive ratio of $0.688$.
\end{proof}



\section{Proof of \Cref{lem:verification_687} and \Cref{lem:verification_689} with Computer Assistance}
\label{app:program-verify}
Below we prove \Cref{lem:verification_689}. \Cref{lem:verification_687} can be proved in a similar way except that the function $\Gamma$ has one less input $s$ which is fixed to be $2$.

First, we set the parameter $s$ as the following:
\[
s = \begin{cases}
3 & x_0 \in [0,0.35] \\
2.5 & x_0 \in (0.35,0.6] \\
2 & x_0 \in (0.6,1]
\end{cases}
\]
Next, let $\epsilon = 1/10000$.
For every $x_0\in [0,1]$, and $h_0\in [0,x_0]$, we round them down to multiples of $\epsilon$. I.e., let $x_0' = \epsilon \cdot \lfloor x_0 / \epsilon \rfloor$ and $h_0' = \epsilon \cdot \lfloor h_0 / \epsilon \rfloor$.


Next, we prove that for any $0\le \beta_0 \le \beta_1 \le \beta_2 \le 1$, we have:
\begin{equation}
\label{eqn:error}
    \Gamma(x_0, h_0, s, \beta_0, \beta_1, \beta_2) \ge \Gamma(x_0', h_0', s, \beta_0, \beta_1, \beta_2) - \left(\frac{3}{2} s^2 + \frac{1}{2} s \right) \epsilon.
\end{equation}

We naturally consider variables according to the rounding: $h_{ot}' = \min\{1-x_0', h_s(x_0') - h_0'\}$ and $K(t)' = h_ot' \cdot t + s(1-x_0'-h_{ot}')(t-\beta_1)^+$, and $L(t)' = h_0' \cdot (t-\beta_0)^+ + s(x_0' - h_0')(t-\beta_2)^+$.
 
We first bound $K(t)' - K(t)$ for any $t$:
\begin{multline*}
K(t)' - K(t) = s(t-\beta_1)^+ \cdot (x_0 - x_0') + \left(s(t-\beta_1)^+ - t\right) \cdot (h_{ot} - h_{ot}') \\
\ge \left(s(t-\beta_1)^+ - t\right) \cdot (h_{ot} - h_{ot}') \ge - \frac{s+1}{2} \epsilon~,
\end{multline*}
where the last inequality holds by the following case analysis.
\begin{itemize}
\item If $s(t-\beta_1)^+ - t \ge 0$, then $s(t-\beta_1)^+ - t \le s-1$. Further by $h_{ot} \ge h_{ot}' - \epsilon$ due to $1 - x_0 \ge 1 - x_0' - \epsilon$ and $h_s(x_0) - h_0 \ge h_s(x_0') - h_0' - \epsilon$, we have $K(t)' - K(t) \ge -(s-1)\epsilon \ge -\frac{s+1}{2}\cdot \epsilon$.
\item If $s(t-\beta_1)^+ - t\le 0$, then $t - s(t-\beta_1)^+ \le 1$. Further by $h_{ot'} \ge h_{ot}' - \frac{s+1}{2}\epsilon$ due to $1-x_0' \ge 1-x_0$ and $h_s(x_0') - h_0' \ge h_s(x_0') - h_0 \ge h_s(x_0) - h_0 - \frac{s+1}{2}\epsilon$ by \Cref{lem:h-derivative}, we have $K(t)' -  K(t) \ge -\frac{s+1}{2}\epsilon$.
\end{itemize}

Second we lower bound $L(t)' - L(t)$ for any $t$:
\[
    L(t)' - L(t) = -s(t-\beta_0)^+ \cdot (x_0 - x_0') + ((t-\beta_0)^+ + s(t-\beta_2)^+)\cdot (h_0 - h_0') \ge -s\cdot \epsilon.
\]
Putting these together, we have for any $t$:
\begin{align*}
    & e^{-K(t)} - e^{-K'(t)} \ge -\frac{s+1}{2} \cdot \epsilon, \\
    & e^{-K(t)}(1-L(t)) - e^{-K(t)'}(1-L(t)') \ge -\left(\frac{s+1}{2} + s\right) \cdot \epsilon.
\end{align*}
Recall the definition of $\Gamma$ function, we conclude the proof of \eqref{eqn:error}.

Finally, we are left to verify that for every $(x_0,h_0) = (i \epsilon, j \epsilon)$, $0 \le i,j \le \frac{1}{\epsilon}$, there exist $0\le \beta_0 \le \beta_1 \le \beta_2 \le 1$,
\[
    \Gamma(x_0, h_0, s, \beta_0, \beta_1, \beta_2) - s\left(\frac32 s + \frac12\right) \cdot \epsilon > 0.688~.
\]
This final step is done by brute force search with computer assistance. For verification, the source code can be found at a \hyperlink{https://github.com/billyldc/prophet_secretary_and_matching_code}{github repository}\footnote{https://github.com/billyldc/prophet\_secretary\_and\_matching\_code}.

\end{document}